\documentclass[11pt,reqno]{amsart}
\usepackage{amsaddr}
\usepackage[T1]{fontenc}
\usepackage[utf8]{inputenc}
\usepackage{soul,todonotes}
\usepackage{amsfonts, amssymb, amsmath, amsthm,latexsym,slashed}
\usepackage{mathtools, thm-restate}
\usepackage{stmaryrd}
\usepackage{hyperref}
\usepackage{amsrefs}
\usepackage{a4wide}
\usepackage{enumitem}
\usepackage{mathrsfs}
\usepackage{bbm}
\usepackage{upgreek} 
\usepackage{nicefrac}
\usepackage{tikz} 
\usetikzlibrary{arrows}
\usepackage{caption}
\usepackage{ulem} 
\usepackage{leftidx}
\hypersetup{
    colorlinks,
    linkcolor={green!50!black}, 
    citecolor={blue}, 
    urlcolor={cyan!75!blue} 
}
\newcommand{\rmi}{\mathrm{i}}
\newcommand{\der}{\mathrm{d}}

\newcommand{\End}[1]{\mathrm{End} ({#1})}

\newcommand{\R}{\mathbb{R}}
\newcommand{\C}{\mathbb{C}}

\newcommand{\one}{\mathbbm{1}}

\newcommand{\supp}[1]{\mathrm{supp} \left({#1}\right)}

\renewcommand{\Box}{P}
\newcommand{\WF}[1]{\mathrm{WF} ({#1})}
\newcommand{\WFPrime}[1]{\mathrm{WF}^{\prime} ({#1})}
 
\newcommand{\Feyn}{\mathrm{F}}
\newcommand{\aFeyn}{\mathrm{aF}}
\newcommand{\ret}{\mathrm{ret}}
\newcommand{\adv}{\mathrm{adv}}

\newcommand{\secE}{C^{\infty} (M; E)}

\newcommand{\ms}{\scriptscriptstyle}

\newcommand{\re}{\mathrm{e}}

\newcommand{\xxi}{(x, \xi)}

\newcommand{\cA}{\mathcal{A}}
\newcommand{\cC}{\mathcal{C}}
\newcommand{\cD}{\mathcal{D}}

\newcommand{\cF}{\mathcal{F}}
\newcommand{\cM}{\mathcal{M}}
\newcommand{\cR}{\mathcal{R}}
\newcommand{\cU}{\mathcal{U}}
\newcommand{\cV}{\mathcal{V}}

\newcommand{\symb}[1]{\sigma_{\ms{#1}}}

\newtheorem{theorem}{Theorem}
\newtheorem{lemma}[theorem]{Lemma}
\newtheorem{corollary}[theorem]{Corollary}

\theoremstyle{definition}

\newcommand{\sadj}[1]{\leftidx{^{\star}}{#1}{}}
\newcommand{\dlangle}{\langle\!\langle}
\newcommand{\drangle}{\rangle\!\rangle}

\newcommand{\PJ}{\mathrm{PJ}}
\newcommand{\id}{\mathrm{id}}

\allowdisplaybreaks
\title[Feynman Propagators]{On the construction of Hadamard states\\ from Feynman propagators}
\author[C.J.~Fewster]{Christopher J.~Fewster}
\address{Department of Mathematics, Ian Wand Building, Deramore Lane,
University of York,\\ York YO10 5GH, UK.}
\email{chris.fewster@york.ac.uk}
\author[A.~Strohmaier]{Alexander Strohmaier}
\address{Institut f\"{u}r Analysis, Leibniz Universit\"{a}t Hannover, Welfengarten 1, 30167 Hannover, Germany.}
\email{a.strohmaier@math.uni-hannover.de}
\subjclass[2010]{}
\date{2 October 2025. Revised, 20 April 2026.}
\begin{document}
\begin{abstract}
The Wightman two-point function of any Hadamard state of a linear quantum field theory determines a corresponding Feynman propagator.
Conversely, however, a Feynman propagator determines a state only if certain positivity conditions are fulfilled. Choosing a Feynman propagator to satisfy the correct positivity conditions involves a slightly subtle point that we address and resolve. Starting from a recent
generalisation of the Duistermaat--H\"ormander theory of distinguished parametrices to normally hyperbolic and Dirac-type operators acting on sections of hermitian vector bundles, we complete this work by showing how Feynman propagators can be chosen so as to define Hadamard states.
The theories considered are: the complex bosonic field governed by a normally hyperbolic operator; the corresponding hermitian theory if the operator commutes with a complex conjugation; the Dirac fermionic theory governed by a Dirac-type operator, and the corresponding Majorana theory in the case where the operator commutes with a skew complex conjugation. The additional key ingredients that we supply are simple domination properties of self-adjoint smooth kernels.
\end{abstract}
\maketitle

\section{Introduction and statement of main results}

The existence of Hadamard states for bundle-valued quantum fields on curved spacetimes is -- as first realised by Radzikowski~\cite{Radzikowski_CMP_1996} -- deeply connected to the theory of distinguished parametrices originally developed by Duistermaat and H\"ormander~\cite{DuiHoer_FIOii:1972}. A recent work~\cite{IslamStrohmaier:2020} has generalised this theory to the vector-bundle context and shown how Feynman propagators can be constructed by microlocalization methods for theories governed by normally hyperbolic and Dirac type operators obeying suitable conditions. Here, we complete that study by showing how Hadamard states of these theories can be constructed using the main results in~\cite{IslamStrohmaier:2020}, which require some additional generalisations for this task. General references on Hadamard states can be found in~\cite{IslamStrohmaier:2020} and also in~\cite{Fewster:2025a}, which describes their generalisation beyond the classes of normally hyperbolic and Dirac type operators, in which such operators still appear frequently as auxiliary objects or particular cases.

Suppose that $(M,g)$ is an $n \geq 2$-dimensional globally hyperbolic spacetime, i.e., an oriented and time-oriented Lorentzian manifold $M$ of metric $g$-signature $(+, -, \ldots, -)$ that possesses a smooth spacelike Cauchy hypersurface.  
Let $E \to M$ be a smooth complex vector bundle  equipped with a nondegenerate but not necessarily positive definite hermitian inner product on its fibres.
A linear differential operator $\Box$ acting on smooth sections $\secE$ of $E$ is called \textit{normally hyperbolic} (d'Alembert-type) if its principal symbol $\symb{\Box}$ is scalar and equals $-g^{-1}$ on the cotangent bundle $T^{*} M$ over $M$, i.e., 
\begin{equation*}
	\symb{\Box} \xxi := -g_{x}^{-1} (\xi, \xi) \, \one_{\End{E}}   
\end{equation*} 
for any $\xxi \in  T^{*} M$.  Such operators are of the form $\square_{\nabla} + V$, where
$\square_{\nabla}$ is the connection d'Alembert operator with respect to some connection on $E$ and $V \in C^\infty(M;\mathrm{End}(E))$ is a smooth potential. The prototype is the Minkowski operator $\Box = \partial^2/\partial t^2 - \triangle + V$ in inertial coordinates. 

Under these assumptions, $\Box$ has a variety of fundamental solutions, also called propagators, i.e., continuous inverses of $\Box$ between various function spaces. In particular, 
$\Box$ has unique retarded and advanced fundamental solutions
$$
G_{\ret} : C^\infty_0(M;E) \to C^\infty(M;E), \quad G_{\adv} : C^\infty_0(M;E) \to C^\infty(M;E)
$$
that map forward and backward in time in the sense that
$$
\mathrm{supp}(G_{\ret/\adv} f) \subset J^{\pm}( \mathrm{supp}( f) ).
$$
However, $\Box$ also admits (non-unique) Feynman and anti-Feynman propagators that satisfy specific conditions on the wave-front set of their distributional kernels. Here, we recall that any continuous
map $C^\infty_0(M;E) \to C^\infty(M;E)$ has a unique integral kernel in $\mathcal{D}'(M\times M; E \boxtimes E^*)$ by the Schwartz kernel theorem. To minimise notation, here and throughout, we use the same symbol for a map and its integral kernel and also use the metric volume density to identify functions and densities.

The fundamental solutions play important roles in the theory of a complex bosonic linear quantum field based on the equation $\Box \phi=0$, where we now assume that $\Box$ is formally self-adjoint. In particular, the operator $G=G_\adv-G_\ret$ determines the canonical commutation relations of the theory. Meanwhile, given a Hadamard state, the expectation value of the smeared field products $\Phi(u)\Phi^\star(f)$ determines a two-point function $W_{\Phi\Phi^\star}\in\mathcal{D}'(M\times M;E\boxtimes E^*)$, from which a Feynman propagator may be obtained by the formula
$$
G_{\Feyn} = \rmi W_{\Phi\Phi^\star} + G_\adv.
$$
Turning this around, one may attempt to construct a Hadamard state from a Feynman propagator by solving the above formula for $W_{\Phi\Phi^\star}$. As the probability interpretation of quantum theory requires that $W_{\Phi\Phi^\star}$ should have positive type, a necessary condition on $G_\Feyn$ for this construction to succeed is that 
\begin{equation}\label{eq:pos1}
-\rmi (G_\Feyn-G_\adv) \ge 0.
\end{equation}
That this condition can be satisfied for scalar parametrices (i.e., mod $C^\infty$) is a result of Duistermaat and H\"ormander \cite{DuiHoer_FIOii:1972}, and it can be improved to hold exactly for propagators as for example shown in \cite{IslamStrohmaier:2020}, which also generalises the results to the bundle setting and streamlines the argument. However, this is not the whole story: there is a second two-point function $W_{\Phi^\star\Phi}\in\mathcal{D}'(M\times M;E^*\boxtimes E)$, determined by field products $\Phi^\star(f)\Phi(u)$, and which is related to $W_{\Phi\Phi^\star}$ by the canonical commutation relations,
$$
W_{\Phi\Phi^\star}(u\otimes f) -W_{\Phi^\star\Phi}(f\otimes u)=\rmi G(u\otimes f).
$$
Consequently, $W_{\Phi^\star\Phi}=-\rmi (G_\Feyn-G_\ret)$. As $W_{\Phi^\star\Phi}$ also has positive type, we find a second positivity condition
\begin{equation}
	\label{eq:pos2}
-\rmi (G_\Feyn-G_\ret) \ge 0.
\end{equation}
If both positivity conditions are satisfied simultaneously, then one may indeed construct a quasifree Hadamard state from $G_\Feyn$.
The question we address and resolve in this short paper is whether any Feynman propagator obeying the first positivity condition may be modified so as to produce a Feynman propagator satisfying both conditions at once. The key to this is a simple domination lemma for smooth kernels that will be proved in Section~\ref{sec:domino} and may be of independent interest.

Assuming that the fibre metric of $E$ is positive definite, our result  immediately provides a proof of the existence of Hadamard states for the complex quantum field theory (QFT) determined by a formally self-adjoint normally hyperbolic operator $\Box$, using the recent results of~\cite{IslamStrohmaier:2020}, which streamlines and generalises to the bundle case the classic work of Duistermaat and H\"ormander~\cite{DuiHoer_FIOii:1972} on the existence of Feynman propagators obeying positivity condition~\eqref{eq:pos1}. In this way the present paper completes the work of~\cite{IslamStrohmaier:2020}. 

In the case where $\Box$ commutes with an antilinear conjugation on $E$, one may define an associated hermitian bosonic QFT as well as the complex theory described above. To obtain a Hadamard state one requires not only the positivity condition~\eqref{eq:pos1} but also that $G_\Feyn$ obeys a symmetry condition with respect to a bilinear form on the sections of $E$. Again, we show that these conditions can be satisfied.
This is stronger than various other positivity conditions considered for example in  \cites{MR3611021,Lewandowski_JMP_2022}.
To construct a Hadamard state the positivity conditions must be verified on all complex-valued sections, even in the case of a hermitian theory. This point was overlooked in \cite{Lewandowski_JMP_2022} and the present paper can be regarded as bridging the gap.

The natural analogue of normally hyperbolic operators for fermionic theories is the class of Dirac-type operators, or, where a compatible skew conjugation is available, a class of Majorana-type operators introduced here. 
Building on~\cite{IslamStrohmaier:2020}, we again establish that Feynman propagators may be modified so as to define Hadamard states, provided the Dirac-type or Majorana-type operator is of definite type.
Here, the domination lemma is used in conjunction with a further property of properly supported smooth kernels.

In the rest of this introduction, we explain the general setting and our main results for normally hyperbolic operators in greater detail, also establishing some notation and terminology.  
Propagators of $\Box$ are partly classified by the wavefront sets of their integral kernels, using the Fourier transform convention  $\hat{f}(k)=\int \re^{-\rmi k\cdot x}f(x) \der x$ with $\cdot$ indicating the Euclidean inner product.
In fact, it is convenient to consider not their wavefront sets 
$\WF K$ but instead the primed wavefront set 
$$
\WFPrime{K} = \{(x,\xi,x',\xi')  \mid (x,\xi,x',-\xi') \in \WF K\}.
$$ 
We will refer to this as the wavefront relation. Continuous inverses that are only inverses modulo smoothing operators will as usual be called parametrices. 

The wavefront sets distinguish parametrices uniquely in the sense that all parametrices with given wavefront set relation coincide up to the addition of a smoothing operator. In particular, the retarded and advanced fundamental solutions have wavefront relations
$$
\WFPrime{G_{\ret/\adv}} = \WFPrime{\id} \cup C_{\ret/\adv}, 
$$
where 
\begin{align*}
	\WFPrime{\id} 
	&:= 
	\{(x,\xi;x,\xi) \in \dot T^*(M \times M) \mid (x,\xi)\in\dot T^*M\},\notag\\
	C_{\ret}
	&:=
	\{(x,\xi;x',\xi') \in \dot T^*(M \times M) \mid \xi \in \dot T^*_0 M, \; x \in J^+(x'),\; \exists\, t \in \R:  (x,\xi)=\Phi_t(x',\xi') \}, \notag\\
	C_{\adv}
	&:=
	\{(x,\xi;x',\xi') \in \dot T^*(M \times M) \mid \xi \in \dot T^*_0 M,\; x \in J^-(x'),\; \exists\, t \in \R:  (x,\xi)=\Phi_t(x',\xi')  \} \notag
\end{align*}
and other symbols are defined below.
Meanwhile, Feynman parametrices $G_\Feyn$ and anti-Feynman parametrices $G_\aFeyn$ are uniquely determined as parametrices (modulo smoothing operators) by their wavefront relation
$$
\WFPrime{G_{\Feyn/\aFeyn}} = \WFPrime{\id} \cup C_{\Feyn/\aFeyn}, 
$$
where
\begin{align*}
	C_\Feyn
	&:=
	\{(x,\xi;x',\xi') \in \dot T^*(M \times M) \mid \xi \textrm{ is lightlike,}\; \exists\, t>0: (x,\xi) =\Phi_{-t}  (x',\xi') \}, \\
	C_\aFeyn
	&:=
	\{(x,\xi;x',\xi') \in \dot T^*(M \times M) \mid \xi \textrm{ is lightlike,}\; \exists\, t>0: (x,\xi) = \Phi_{t} (x',\xi') \},
\end{align*}
and a Feynman/anti-Feynman propagator $G_{\Feyn/\aFeyn}$ is a fundamental solution
$$
G_\Feyn : C^\infty_0(M;E) \to C^\infty(M;E)
$$
of $\Box$ that has the wavefront relation $\WFPrime{G_{\Feyn/\aFeyn}}$. 

In the above formulae, $\dot T^*_0 M$ denotes the set of non-zero null-covectors, while $\Phi_t$ denotes the time $t$ geodesic flow. 
The null covectors appear here because $\dot T^*_0 M$ is the characteristic set of $\Box$.
It decomposes into a disconnected union of future and past light-cones $\dot T^*_0 M = \dot T^*_{0,+} M \cup \dot T^*_{0,-} M $, where a null covector is regarded as future-pointing if
it has positive contraction with any future-directed timelike vector. 
Given our signature convention, covectors in $T^*_{0,+}M$ correspond to future-pointing null vectors under the musical isomorphism determined by the metric (see Appendix~\ref{appx:notation} for further comments).
We use the common notation that
$\dot T^* M$ is the cotangent bundle with its zero section removed.

Any fundamental solution necessarily has to contain the canonical relation of the identity map $\WFPrime{\id}$. One way to understand the wavefront set relations precisely is to use the fact that the operator $\Box$ near any point of its characteristic set $(x,\xi)$ is microlocally equivalent to the operator $\rmi \partial_n =\rmi  \frac{\partial}{\partial x_n}$ on the space $C^\infty_0(\R^n,\C^N)$ where $N$ is the rank of $E$. The latter operator has two distinguished fundamental solutions, forward and backward ones.
A parametrix for $\Box$ can then be constructed as follows. One microlocally conjugates the operator $\Box$ near points of its characteristic set to $\rmi \partial_n$ using Fourier integral operators. This conjugation is used to construct a microlocal parametrix of $\Box$. These microlocal parametrices are patched together to obtain a full parametrix using a microlocal partition of unity.
It is clear that this can be done by picking forward or backward fundamental solutions for  $\rmi \partial_n$ on each connected component of the characteristic set of $\Box$. In case of the two components of the light cone (in case $n\geq 3$) this gives us four possibilities for the construction of such distinguished parametrices. For example, the retarded fundamental solution is constructed by choosing the forward fundamental solution for $\rmi \partial_n$ in the forward light cone, but the backward fundamental solution in the backward light cone. This results in all singularities propagating to the future, thus creating the wavefront relation $C_{\ret}$.  Similarly, by choosing the backward fundamental solution on the forward light cone and the forward solution on the backward light cone, one obtains $C_{\adv}$.
The other canonical choice is to pick the forward/backward
fundamental solutions throughout resulting in the anti-Feynman and Feynman parametrices. 
 
  Now write $G\sim G'$ if $G-G'$ is a smoothing operator.
Assuming that we are given parametrices $G_\ret,G_\adv, G_\Feyn,G_\aFeyn$ as above it is then clear from the construction that 
\begin{align*}
 G_\Feyn &\sim T_{--} + T_{-+},&
 G_\aFeyn &\sim T_{++} + T_{+-},\\
 G_\ret &\sim T_{++} + T_{--},&
 G_\adv &\sim T_{-+} + T_{+-},
\end{align*}
where $T_{++}$ has been constructed from the forward fundamental solution on the future light cone, $T_{+-}$ from the forward fundamental solution on the past light cone, $T_{-+}$ from the backward fundamental solution on the future light cone, and
$T_{--}$ from the backward fundamental solution on the past light cone.
Thus, for $\sigma,\sigma'\in\{+,-\}$, 
$$
\WFPrime{T_{\sigma\sigma'}}\subset \WFPrime{\id}\cup 
\{(x,\xi;x',\xi') \in \dot T^*(M \times M) \mid \xi' \in \dot{T}^*_{0,\sigma'}\; \exists\, t>0: (x,\xi) = \Phi_{ \sigma t}(x',\xi') \}.
$$

One sees from this that $G_\Feyn + G_\aFeyn \sim G_\ret + G_\adv$ for example. More importantly,
$$
G_\adv - G_\ret  = 
(G_\Feyn-G_\ret)- (G_\Feyn-G_\adv) \sim (T_{-+}-T_{++} )  - (T_{--} - T_{+-})
$$
achieves a splitting of the operator $G =  G_\adv - G_\ret$, with wavefront relation in  
$$
C = \{(x,\xi;x',\xi') \in \dot T^*(M \times M) \mid \xi \in T^*_0 M ; \exists\, t \in \R: (x,\xi) = \Phi_{ t}(x',\xi') \},
$$
into operators $T_{-\pm} -T_{+\pm} \sim G_\Feyn-G_{\ret/\adv}$ with wavefront relation in 
$$
C_\pm = \{(x,\xi;x',\xi') \in \dot T^*(M \times M) \mid \xi \in T^*_{0,\pm} M ; \exists\, t \in \R: (x,\xi) =\Phi_{ t} (x',\xi') \}.
$$
Since $C = C_+ \cup C_-$ and $C_+\cap C_- = \emptyset$ such a splitting is unique modulo smooth kernels. By definition, a state of the complex QFT determined by $\Box$ is Hadamard if and only if its two-point functions satisfy
\begin{equation}\label{eq:Hadcov}
\WFPrime{W_{\Phi\Phi^\star}} = \WFPrime{W_{\Phi^\star\Phi}} = C_-,
\end{equation}
so the previous observation on the uniqueness of the splitting corresponds to the fact that all Hadamard two-point functions share a common singular structure.

Returning to Feynman propagators, a
key result, by Duistermaat and Hörmander in the scalar case~\cite{DuiHoer_FIOii:1972}, and Islam and Strohmaier~\cite{IslamStrohmaier:2020} in the case of hermitian vector bundles, is that they can be chosen so that additional self-adjointness and positivity properties are preserved.
\begin{theorem}[Duistermaat--Hörmander~\cite{DuiHoer_FIOii:1972}, Islam--Strohmaier~\cite{IslamStrohmaier:2020}] \label{thm:DHIS}
 Assume that $\Box$ is formally self-adjoint  and the bundle inner product of $E$ is positive definite. Then there exists a Feynman propagator $G_\Feyn$ 
 such that
 \begin{align*}
   -\rmi (G_\Feyn -G_\adv) \geq 0.
 \end{align*}
\end{theorem}

It is also stated in \cite{IslamStrohmaier:2020} that $-\rmi (G_\Feyn -G_\adv)$ defines a Hadamard state, but without any further explanation. While this bidistribution is of positive type and satisfies the equation of motion, for it to be a state on the $*$-algebra obtained by canonical quantisation of the complex or real scalar field further conditions need to be satisfied.  
In fact, as will be described in more detail later, a gauge-invariant quasifree 
state of the complex theory is determined by two distributions $W_{\Phi\Phi^\star}\in\cD'(M\times M;E\boxtimes E^*)$ and 
$W_{\Phi^\star\Phi}\in\cD'(M\times M;E^*\boxtimes E)$ obeying
\begin{align}\label{eq:covariances}
	(\Box\otimes 1)W_{\Phi\Phi^\star} &= (1\otimes \sadj{\Box})W_{\Phi\Phi^\star} =0 \nonumber \\
	(1\otimes \Box)W_{\Phi^\star\Phi} &= (\sadj{\Box}\otimes 1)W_{\Phi^\star\Phi} =0 \nonumber \\
	W_{\Phi\Phi^\star}(u\otimes f) - W_{\Phi^\star\Phi}(f\otimes u) &= \rmi\dlangle u, Gf\drangle  \\
	W_{\Phi\Phi^\star}(f^\star\otimes f)& \ge 0 \nonumber \\ 
	W_{\Phi^\star\Phi}(f\otimes f^\star)&\ge 0\nonumber
\end{align}
for all test-sections $f\in C_0^\infty(M;E)$ and $u\in C_0^\infty(M;E^*)$. Here, $\sadj{\Box}$, acting on $C^\infty(M;E^*)$, is the formally dual operator to $\Box$, 
and $f\mapsto f^\star$ is the conjugate-linear isomorphism of $E$ and $E^*$ determined fibrewise by
$\dlangle f^\star,h\drangle_x = \langle f,h\rangle_x$, in which the single and double brackets indicate the hermitian and duality pairings at $x$. All but the last of the requirements in~\eqref{eq:covariances} are easily seen to follow from Theorem~\ref{thm:DHIS} if one sets
\begin{equation}\label{eq:covdefs}
	W_{\Phi\Phi^\star}(u\otimes f)=-\rmi\langle u^\star,(G_\Feyn-G_\adv) f\rangle, \qquad 
	W_{\Phi^\star\Phi}(f\otimes u)=-\rmi\langle u^\star,(G_\Feyn-G_\ret) f\rangle,
\end{equation}
for $f\in C_0^\infty(M;E)$, $u\in C_0^\infty(M;E^*)$. Thus, to obtain a state one must also require simultaneously that $-\rmi (G_\Feyn -G_\ret) \geq 0$.
This inequality is however not an automatic consequence of
Theorem~\ref{thm:DHIS}.

In this paper, we show that a relatively simple lemma about the domination of self-adjoint smooth kernels by positive smooth kernels (which may be interesting in its own right) can be used to show that in fact both positivity relations can consistently be enforced.
\begin{theorem}\label{thm:DHISmod}
 Assume that $\Box$ is formally self-adjoint and the bundle inner product of $E$ is positive definite. Then there exists a Feynman propagator $G_\Feyn$ 
 such that
 \begin{align*}
   -\rmi (G_\Feyn -G_\ret) \geq 0,\\
  -\rmi (G_\Feyn -G_\adv) \geq 0.
 \end{align*}
\end{theorem}
To keep the exposition self-contained we will also sketch the relevant parts of the construction of \cite{DuiHoer_FIOii:1972} and \cite{IslamStrohmaier:2020} when proving Theorem~\ref{thm:DHISmod}.
With Theorem~\ref{thm:DHISmod} in hand, we can construct gauge-invariant quasifree states of the QFT using~\eqref{eq:covdefs}, which moreover fulfill the Hadamard condition~\eqref{eq:Hadcov}: 
\begin{corollary}\label{cor:complex}
	Let $\Box$ be any formally self-adjoint normally hyperbolic operator on smooth finite-rank complex positive-definite hermitian vector bundle $E\to M$ over a globally hyperbolic spacetime $M$. 
	Then the quantised complex field with equation of motion operator $\Box$ admits a gauge-invariant quasifree Hadamard state, whose two-point functions are determined by the
	Feynman propagator provided by Theorem~\ref{thm:DHISmod}.
\end{corollary}
This statement completes and makes explicit Theorem~1.4 in~\cite{IslamStrohmaier:2020} for complex fields.
However, it does not automatically yield a Hadamard state for the real scalar field or general fields with a real structure. Assume that our complex positive-definite hermitian vector bundle has an additional structure, namely a complex conjugation (a fibrewise conjugate-linear involution) $\overline \cdot: E \to E$ that commutes with the operator $\Box$, which also determines a linear isomorphism $f\mapsto f^T$ between (sections of) $E$ and $E^*$ such that
$\dlangle f^T,h\drangle = \langle \overline{f},h\rangle$.
 In this case we can define the anti-symmetric Pauli-Jordan commutator distribution
{$G_\PJ\in \mathcal{D}'(M \times M; E^* \boxtimes E^*)$ from $G$ using the conjugation, setting
$$
G_\PJ(f\otimes h) = \langle \overline{f}, G h\rangle = G(f^T\otimes h). 
$$}
In order for a bi-distribution  $W \in \mathcal{D}'(M \times M; E^* \boxtimes E^*)$ to define a Hadamard state on the CCR-algebra for the hermitian QFT based on $\Box$, one requires
\begin{align}\label{eq:Hcovariance}  
	(\sadj{\Box}\otimes 1)W &= (1\otimes \sadj{\Box})W = 0 \nonumber\\
	W(\overline{f}\otimes f) &\ge 0 \\
	W(f\otimes h)- W(h\otimes f) &=\rmi G_\PJ(f\otimes h)   \nonumber
\end{align} 
and that $\WFPrime{W}\subset C_-$. Whereas the wavefront condition and the first two parts of~\eqref{eq:Hcovariance} are met by setting
\begin{equation}\label{eq:Wdef}
W(f\otimes h) = -\rmi(G_\Feyn-G_\adv)(f^T\otimes h)
\end{equation}
with $G_\Feyn$ furnished by Theorem~\ref{thm:DHIS}  or Theorem~\ref{thm:DHISmod}, the last condition in~\eqref{eq:Hcovariance} does not follow
automatically: the uniqueness of parametrices only guarantees that it will hold modulo smoothing operators. An exact equality would require one to modify the patching construction used in the proof of Theorem~\ref{thm:DHIS} to be compatible with the complex conjugation. Instead, we will show that a variant of the domination lemma allows one to modify $G_\Feyn$ by a smooth kernel so that the third condition is satisfied as well.

\begin{theorem}\label{thm:improved_Feynman}
  Assume that $\Box$ is formally self-adjoint and real,  and the bundle inner product is positive-definite. Then there exists a Feynman propagator $G_\Feyn$ such that
 \begin{align*}
   -\rmi (G_\Feyn -G_\ret) \geq 0,\\
  -\rmi (G_\Feyn -G_\adv) \geq 0,
 \end{align*}
 and 
 $$
  \langle \overline{f}, G_\Feyn h \rangle = \langle \overline{h}, G_\Feyn f \rangle.
 $$ 
\end{theorem} 
{Indeed, a suitable Feynman propagator may be obtained by symmetrising one obtained from Theorem~\ref{thm:DHISmod}, for example. Theorem~\ref{thm:improved_Feynman}} then implies that the two-point function in~\eqref{eq:Wdef} obeys $W(f \otimes h) - W(h \otimes f) = \rmi G_{\PJ}(f\otimes h)$. Summarising:
\begin{corollary}\label{cor:hermitian} 
	Let $\Box$ be any real formally self-adjoint normally hyperbolic operator on smooth finite-rank complex  positive-definite hermitian vector bundle $E\to M$ over a globally hyperbolic spacetime $M$, assuming that $E$ is equipped with a complex conjugation. Then the quantised hermitian field with equation of motion operator $\Box$ admits a quasifree Hadamard state, whose two-point function is determined by the
	Feynman propagator provided by Theorem~\ref{thm:improved_Feynman}.
\end{corollary}
This result is used in the proof of Theorem~5.3 in~\cite{Fewster:2025a}, which develops the Hadamard condition for a class of Green hyperbolic operators going beyond the normally hyperbolic class but reducing to the usual notion of Hadamard in that situation. 

We proceed as follows. We  establish the domination lemma in Section~\ref{sec:domino} and then prove Theorems~\ref{thm:DHISmod} and~\ref{thm:improved_Feynman} in Section~\ref{sec:proof} before elaborating in greater depth on the algebraic quantisation of real and complex fields associated with $\Box$ in Section~\ref{sec:algebras}. 
This has two purposes. First it provides the background for the conditions~\eqref{eq:covariances} and~\eqref{eq:Hcovariance}. Second it also leads to  an alternative proof of  Corollary~\ref{cor:hermitian},
using Corollary~\ref{cor:complex} and results proved in~\cite{Fewster:2025a} concerning tensor products and partial traces of Hadamard states. This results in an explicit formula 
$$
	W(f\otimes h) = \frac{1}{2}\left(
W_{\Phi^\star\Phi}(f\otimes h^T+h\otimes f^T) +\rmi G_\PJ(f\otimes h)\right) 
$$ 
for a quasifree Hadamard two-point function of the hermitian field in terms of the two-point functions~\eqref{eq:covdefs} of the complex field obtained from Theorem~\ref{thm:DHISmod}. 
Section~\ref{sec:Dirac} proves parallel results for Dirac and Majorana type operators of definite type.
Finally, as this paper uses slightly different conventions to both of~\cites{IslamStrohmaier:2020,Fewster:2025a}, we provide a comparison in Appendix~\ref{appx:notation}.

To close the introduction we mention that Hadamard states can be constructed by other means and there is by now a wealth of methods available in various contexts (see, e.g.,~\cites{FullingSweenyWald:1978,Moretti:2008,GerardOulghaziWrochna:2017,VasyWrochna:2018}).
More generally other types of Feynman propagators without positivity assumptions are associated to scattering processes as in the in-out formalism (see e.g. \cite{dergas2025prop}). The purpose of 
our paper is to provide a general tool-box to convert microlocal splittings into Hadamard states. In particular, in combination
with \cite{IslamStrohmaier:2020} this gives a unified approach applicable to large classes of bundle-valued fields of both bosonic and fermionic nature.
 
\section{A domination property for smooth kernels} \label{sec:domino} 

We now assume $E$ is a complex hermitian vector bundle.

\begin{lemma}\label{lem:domination}
 Assume $k \in C^\infty(M \times M; E \boxtimes E^*)$ is self-adjoint, i.e. satisfies
 $$
  k(x,y)^* = k(y,x),
 $$
 where $k(x,y)$ is regarded as a linear map $E_y \to E_x$ with adjoint $k(x,y)^*$ with respect to the hermitian form.
 Then there exists a kernel $K \in C^\infty(M \times M; E \boxtimes E^*)$ that 
 \begin{enumerate}
  \item[(i)] $K \geq 0$.
   \item[(ii)] $K -k \geq 0$ in the sense that for all $u \in C^\infty_0(M;E)$ we have 
   $$
    \int_{M \times M} \langle u(x), \left( K(x,y) - k(x,y) \right) u(y) \rangle \der x \der y \geq 0.
   $$   
   \item[(iii)] In case $k$ is properly supported, $K$ can also be chosen properly supported.
 \end{enumerate}
  In case the bundle $E$ also admits a complex conjugation we can arrange $K$ so that in addition to the above it satisfies
 \begin{enumerate}
  \item[(iv)] $K$ is real in the sense that $\overline{K(x,y)} = K(x,y)$, where $\overline{K(x,y)}  v_y = \overline{ K(x,y) \overline v_y} $.
  \end{enumerate}
 \end{lemma}
 \begin{proof}
  (i) \& (ii) The space $C^\infty(M;E)$ is a nuclear Fr\'echet space and the projective tensor product
  $C^\infty(M;E) \otimes_\pi C^\infty(M;E^*)$ equals $C^\infty(M \times M; E \boxtimes E^*)$. By 
  Grothendieck's characterisation of the projective tensor product (\cite{MR1741419}*{Ch III, \S6.4}) we can find null-sequences $(\phi_j),(\psi_j)$ in
  $C^\infty(M;E)$ and a sequence $(\lambda_j) \in \ell_1(\mathbb{N})$ such that
  $$
   k =  \sum_j \lambda_j \phi_j \otimes \psi_j^\star,
  $$
  in other words
  $$
   k(x,y) =  \sum_j \lambda_j \phi_j(x)\psi_j(y)^{\star},
  $$
  in $C^\infty(M \times M; E \boxtimes E^*)$. We assume here without loss that $\lambda_j \geq 0$.
  Then, since $k(x,y)$ is self-adjoint we have
  $$
   k(x,y) =  \sum_j \frac{1}{2} \lambda_j  \left( \phi_j(x)  \psi_j(y)^\star +  \psi_j(x) \phi_j(y)^\star \right).
  $$
  Now define the positive kernel
  $$
   K(x,y) = \frac{1}{2} \sum_j \lambda_j \left( \phi_j(x) + \psi_j(x) \right) \left( \phi_j(y)^\star  + \psi_j(y)^\star  \right)
  $$
  Since the sequences $\phi_j, \psi_j, \phi_j^\star, \psi_j^\star$ converge to zero in $C^\infty$ and $\lambda\in\ell_1(\mathbb{N})$, the sum converges in $C^\infty(M \times M; E \boxtimes E^*)$. Moreover,
  $$
   K(x,y) - k(x,y) =  \frac{1}{2}\left( \sum_j \lambda_j  \phi_j(x) \phi_j(y)^\star \right) +  \frac{1}{2} \left(\sum_j \lambda_j  \psi_j(x) \psi_j(y)^\star \right) 
  $$
  and therefore $K-k \geq 0$.
  
  (iii) Now assume that $k$ is properly supported. It can then be written as a locally finite sum $\sum_\alpha k_\alpha$ of compactly supported smooth kernels $k_\alpha$.
  To see this simply write $k = \sum_\alpha \frac{1}{2} \left( k \psi_\alpha + \psi_\alpha k  \right)$ for a partition of unity $\psi_\alpha$ by compactly supported smooth functions
  and choose $k_\alpha =  \frac{1}{2} \left( k \psi_\alpha + \psi_\alpha k  \right)$, which is compactly supported because $k$ is properly supported.
 Next  choose compact sets $X_\alpha \subset M$ and open sets $\mathcal{U}_\alpha\subset M$ such that $X_\alpha \subset \mathcal{U}_\alpha$ in such a way that (a) each $k_\alpha$ is supported in $X_\alpha \times X_\alpha$ and (b) each compact subset of $M$ is intersected by only finitely many $\mathcal{U}_\alpha$.
  Now choose compactly supported real-valued smooth functions $\varphi_\alpha$ supported in $\mathcal{U}_\alpha$ and equal to one on $X_\alpha$. For each $k_\alpha$ we can then find a smooth positive kernel $K_\alpha$ so that $K_\alpha - k_\alpha \geq 0$. Then
  $$
   K(x,y) = \sum_\alpha \varphi_\alpha(x) K_\alpha(x,y)  \varphi_\alpha(y)
  $$
  is a locally finite sum and defines a smooth positive properly supported kernel $K$.
  Of course
  $$
   K - k = \sum_\alpha \varphi_\alpha \left( K_\alpha - k_\alpha \right)  \varphi_\alpha \geq 0,
  $$ 
   noting that $k_\alpha=\varphi_\alpha k_\alpha \varphi_\alpha$.

   (iv) In case there exists a complex conjugation we replace $K$ by $K + \overline{K}$; {as $\overline{K}\ge 0$, we maintain the domination property, but the overall operator is now real.}
 \end{proof}
 
We will use part~(iii) of Lemma~\ref{lem:domination} in conjunction with an observation about positive kernels. Given a properly supported non-negative smoothing operator $K \in C^\infty(M \times M; E \boxtimes E^*)$ we can think of it as a possibly unbounded operator on $L^2(M;E)$ with dense domain $L^2_\mathrm{comp}(M;E)$, the space of compactly supported $L^2$-sections. The Friedrichs extension of this operator is then a self-adjoint operator on $L^2(M;E)$ and we denote this by the same letter $K$. Indeed, this is again extended by the smoothing operator regarded as a map $\mathcal{D}'(M;E) \to C^\infty(M;E)$.

 \begin{lemma} \label{rootlemma}
  Assume that $K \in C^\infty(M \times M; E \boxtimes E^*)$ is a properly supported non-negative smoothing operator.
  Consider the operator
  $$
   T = \left(\id + K \right)^{-\frac{1}{2}}
  $$
  defined by the spectral calculus of the possibly unbounded operator $\id+K$ in $L^2(M;E)$.
  Then $\id - T$ is a non-negative operator with smooth integral kernel in $C^\infty(M \times M; E \boxtimes E^*)$.
 \end{lemma}
 \begin{proof}
  We can simply write, using spectral calculus,
  $$
   \id - T = \frac{2}{\pi} K \int^\infty_0 (\id+ K + \lambda^2)^{-1} (\id+\lambda^2)^{-1} \der \lambda.
  $$
  Since this is self-adjoint, it will have smooth integral kernel if we can show that this extends as a continuous map
  $H^s_\mathrm{comp}(M) \to C^\infty(M)$ for any $s \in \R$. Since $K$ is properly supported it defines in itself a continuous map
  $H^s_\mathrm{loc}(M)  \to C^\infty(M)$ for any $s \in \R$,  extending the self-adjoint operator $K$.  It is thus sufficient to show that
  $$
   \int^\infty_0 (\id+ K + \lambda^2)^{-1} (\id+\lambda^2)^{-1} \der \lambda
  $$
  defines a continuous map $H^s_\mathrm{comp}(M) \to H^s_\mathrm{loc}(M)$ for all $s \in \R$. By duality it is sufficient to prove that for $s \geq 0$ only.
  Hence, the statement boils down to showing that for any compactly supported smooth function $\chi$ the map $\chi (\id+ K + \lambda^2)^{-1} \chi$
  is a continuous map from $H^s_\mathrm{comp}(M)$ to $H^s_\mathrm{comp}(M)$ which is uniformly bounded in the parameter $\lambda$.
 To see this note the simple identity
 $$
  (\id+ K + \lambda^2)^{-1} = (\id+\lambda^2)^{-1} - K (\id+\lambda^2)^{-2} + K (\id+ K + \lambda^2)^{-1} (\id+\lambda^2)^{-2} K.
 $$
 Since $K \chi$ and $\chi K$ are compactly supported smooth kernels this implies the desired bounds.
 \end{proof}

\section{Proofs of Theorem~\ref{thm:DHISmod} and Theorem~\ref{thm:improved_Feynman}}\label{sec:proof}

We briefly sketch the construction of \cite{DuiHoer_FIOii:1972} and \cite{IslamStrohmaier:2020}.  We use here the construction of the corresponding partition of unity and the Fourier integral operators as a black box but focus on the structural part of the argument that leads to positivity and the relation between the propagators.
Consider the differential operator $\rmi \partial_n$.
If we split the coordinates $(x_1,\ldots,x_{n-1},x_n)=(y,x_n)$ its forward fundamental solution $F_+$ has distributional integral kernel equal to $-\rmi \theta(x_n-x_n') \delta(y-y')$.
The backward fundamental solution $F_-$ has distributional integral kernel equal to $\rmi \theta(-x_n+x_n') \delta(y-y')$. One can see immediately from this that 
$F_+^* = F_-$. Further, since $-\rmi \left( F_- -F_+ \right)$ has integral kernel $\delta(y-y')$ it is a positive distribution and therefore $-\rmi \left( F_- -F_+ \right)\otimes \one_E$ is positive, provided that $E$ is equipped with a positive definite inner product.
The operator $\rmi \partial_n\otimes\one$ serves as a model operator as it is microlocally conjugate to $\Box$ near the characteristic set of $\Box$. 

This is made precise by the construction of a family of pseudo-differential operators $\Psi_\alpha^\pm$,  and Fourier integral operators $A_\alpha^\pm$ that microlocally conjugate $\rmi \partial_n \otimes \one$ to $\Box$.
Here, the $\Psi_\alpha^\pm$ form a microlocal partition of unity subordinate to an open conic cover $\mathcal{U}^\pm_{\alpha}$ of $\dot T_{0,\pm} M$ that does not intersect $\dot T^*_{0,\mp} M$. Thus, the integral kernels of the $\Psi_\alpha^\pm$ have their twisted wavefront set contained in $\{(x,\xi,x,\xi) \mid (x,\xi) \in \mathcal{U}^\pm_{\alpha}\}$.
The $A_\alpha^\pm$ have order $-\frac{1}{2}$ with
$$
\WFPrime{A_\alpha^\pm} \subset  \{(x,\xi,y,\eta)\in\cU^\pm_\alpha \times \cV^\pm_\alpha : (x,\xi)=R^\pm_\alpha(y,\eta)\},
$$
where $R^\pm_\alpha:\cV^\pm_\alpha\to\cU^\pm_\alpha$ is an invertible homogeneous canonical transformation from an open conic subset $\cV^\pm_\alpha\subset \dot T^*\R^n$ to $\cU^\pm_\alpha\subset \dot T^*M$, 
	so that $\{(y,\eta) \in \cV_\alpha^\pm \mid \eta_n =0\}$ is mapped to the light cone part
	$T^*_0M\cap \cU^\pm_\alpha$.  
Further, there is a homogeneous elliptic symbol $e^{\pm}_\alpha$ of order $1$ so that
	$$
	\frac{g_x^{-1}(\xi,\xi)}{e^\pm_\alpha(x,\xi)} = \eta_n
	$$
	if $(x,\xi)=R^\pm_\alpha(y,\eta)$ with $(y,\eta)\in \cV^\pm_\alpha$, and the 
	flow generated by $\partial_n$ on $\{(y,\eta) \in \cV_\alpha^\pm \mid \eta_n =0\}\subset\dot{T}^*\R^n$ is therefore mapped to a rescaled null geodesic flow in $\cU^\pm_\alpha$, which is future-directed on the $\dot{T}^*_{0,+}M$ and past-directed on $\dot{T}^*_{0,-}M$.
	Here the symbol $e^\pm_\alpha$ is necessary to balance the different degrees of the operators $\Box$ and $\rmi \partial_n$.	
	Then,
	$$
	\WFPrime{F_{\pm}}=\WFPrime{\id} \cup \{(y,\eta;y',\eta')\in \dot{T}^*(\R^n\times\R^n):
	y_n \gtrless y_n',~ \underline{y}=\underline{y}',~\eta_n=\eta'_n=0,~\underline{\eta}=\underline{\eta}'
	\}
	$$  
	in which the nondiagonal points are related by the forward/backward flow of $\partial_n$, and the notation is that
	$y=(\underline{y},y_n)$, $\eta=(\underline{\eta},\eta_n)$.
	
This then allows one to construct the microlocal inverses of $\Box$ as
\begin{align*}
  T_{+ \pm} = \sum_\alpha \Psi_\alpha^\pm A_\alpha^\pm ( F_+ \otimes\one) (A_\alpha^\pm)^* (\Psi_\alpha^\pm)^*,\\
  T_{- \pm} = \sum_\alpha \Psi_\alpha^\pm A_\alpha^\pm (F_- \otimes \one ) (A_\alpha^\pm)^* (\Psi_\alpha^\pm)^*,
\end{align*} 
and we see that $\WFPrime{T_{\sigma \sigma'}}$ are as stated in the introduction.
By construction we have that
\begin{align*}
 -\rmi ( T_{--} - T_{+-}) \geq 0,& \quad -\rmi ( T_{-+} - T_{++}) \geq 0,\\
 (T_{+-})^* = T_{--},& \quad (T_{++})^* = T_{-+}.
\end{align*} 
As explained in the introduction one then has that
\begin{align*}
 G_\ret \sim T_{++} + T_{--}, \quad
 G_\adv \sim T_{-+} + T_{+-}.
\end{align*}
{
Therefore
$$
 G_\Feyn = T_{--} + T_{-+} + G_\adv - (T_{-+} + T_{+-}) \sim T_{--} + T_{-+} 
$$
is a Feynman parametrix that has the property that
$$
 -\rmi (G_\Feyn - G_\adv) =  -\rmi( T_{--}  -  T_{+-} ) \geq 0.
$$
This is essentially the proof carried out in \cite{DuiHoer_FIOii:1972} and \cite{IslamStrohmaier:2020}. However, with this choice of $G_\Feyn$ the inequality
$$
 -\rmi (G_\Feyn - G_\ret)  \geq 0
$$
is not automatic since
\begin{align*}
 -\rmi (G_\Feyn - G_\ret) = -\rmi(T_{--} - T_{+-} - (G_\ret - G_\adv)) \sim -\rmi( T_{-+} - T_{++} ).
\end{align*}
Indeed, $-\rmi (G_\Feyn - G_\ret)$ differs from the positive $-\rmi( T_{-+} - T_{++} )$ by a smooth kernel
\begin{align*}
	k &=-\rmi( T_{-+} - T_{++} ) +  \rmi (G_\Feyn - G_\ret) \\ 
	&=  \rmi( G_\adv - G_\ret ) + \rmi (T_{++} + T_{--} - T_{-+} -T_{+-}).
\end{align*}
This kernel is self-adjoint and,  by the domination lemma (Lemma~\ref{lem:domination}), we can therefore find a positive smooth kernel $K$ such that $K - k \geq 0$.
If we now replace $G_\Feyn$ by $G_\Feyn + \rmi K$ we have, on the one hand, that
$$
-\rmi (G_\Feyn - G_\adv) =  -\rmi( T_{--}  -  T_{+-} ) +K\geq 0,
$$
and on the other that 
$$
-\rmi (G_\Feyn - G_\ret) =  -\rmi( T_{-+} - T_{++} ) -k + K \geq 0,
$$
so we have achieved 
\begin{align} \label{feynpos}
	-\rmi (G_\Feyn - G_\ret) \geq 0, \quad  -\rmi (G_\Feyn - G_\adv) \geq 0. 
\end{align}
}

A further standard modification can now be made to replace the parametrix $G_\Feyn$ by an actual fundamental solution.
Here one chooses a spacelike and smooth Cauchy hypersurface $\Sigma$ and solves the distributional Cauchy problem to construct a bisolution $F$ on $M \times M$ whose Cauchy data at $\Sigma\times\Sigma$ agree with that of the kernel of $G_\Feyn - G_\adv$. Then replace 
$G_\Feyn$ by $F+G_\adv$.
Since this operation preserves positivity (\cite{IslamStrohmaier:2020}*{Lemma 3.2}) this gives a Feynman propagator on $M$ satisfying the first inequality in \eqref{feynpos}.
Note here that the operation of passing to Cauchy data at $\Sigma \times \Sigma$ and then extending to a bisolution to $M \times M$ results in the same correction when done with
 $(G_\Feyn - G_\ret)$ because $G_\ret-G_\adv$ is already a strict bisolution.  Therefore the second inequality in~\eqref{feynpos} also holds, thus completing the proof of Theorem~\ref{thm:DHISmod}.
 
It remains to prove Theorem~\ref{thm:improved_Feynman}. 
Define the transpose $A^\mathrm{T}$ of an operator $A: C_{0}^\infty(M,E) \to \mathcal{D}'(M,E)$
	to be an operator from $C_{0}^\infty(M,E)$ to $\mathcal{D}'(M,E)$ obeying
	$$
	\langle \overline f, A h \rangle =  \langle \overline h, A^\mathrm{T} f \rangle,
	$$
	for $f,h\in  C_{0}^\infty(M,E)$.
	Since $\Box$ is real and self-adjoint it is symmetric, i.e., $P^\mathrm{T}=P$, and therefore it follows that
	the retarded and advanced Green operators obey
	\begin{align*}
		G_\ret^\mathrm{T} = G_\adv, \quad G_\adv^\mathrm{T} = G_\ret.
	\end{align*}  
 Let $G_\Feyn$ be a Feynman propagator furnished by Theorem~\ref{thm:DHISmod}. To prove Theorem~\ref{thm:improved_Feynman}, we will use $G_\Feyn$ to construct a
 Feynman propagator that is symmetric, in the sense that
 $G_\Feyn = G_\Feyn^\mathrm{T}$, and also satisfies both positivity conditions. 
 First observe that
 $\widetilde{G}_\Feyn =\frac{1}{2} \left( G_\Feyn + G_\Feyn^\mathrm{T} \right)$ is a symmetric Feynman propagator, $C_\Feyn$ being invariant under the involution $(x,\xi,x',\xi') \mapsto (x',-\xi',x,-\xi)$.
 { 
 Now 
 \[
 -\rmi (\widetilde{G}_\Feyn - G_\ret) = -\frac{\rmi}{2} (G_\Feyn-G_\ret) - \frac{\rmi}{2}(G_\Feyn^\mathrm{T} -G_\adv^\mathrm{T}),
 \]
 and both terms on the right-hand side are nonnegative by Theorem~\ref{thm:DHISmod} and the fact that nonnegativity is preserved by the transpose. Transposing, we also find that 
 $-\rmi (\widetilde{G}_\Feyn - G_\adv)\ge 0$ as required.
 
 An alternative argument is to suppose that $G_\Feyn$ satisfies one but not necessarily both of the positivity properties in Theorem~\ref{thm:DHISmod}, say $-\rmi(G_\Feyn-G_\adv)\ge 0$ as in Theorem~\ref{thm:DHIS}. With $\widetilde{G}_\Feyn$ as before, we have
 \[
-\rmi (\widetilde{G}_\Feyn - G_\adv) = -\rmi (G_\Feyn-G_\adv) + \frac{\rmi}{2}(G_\Feyn-G_\Feyn^\mathrm{T}),
\]
in which the first term on the right-hand side is nonnegative, while the second is self-adjoint with a smooth kernel by the 
uniqueness of distinguished parametrices. Using part~(iv) of the domination lemma, Lemma~\ref{lem:domination}, we may find a positive kernel $K$ that is real, and therefore symmetric, and so that $-\rmi (\widetilde{G}_\Feyn +\rmi K- G_\adv)\ge 0$. Transposing, we find that $-\rmi(\widetilde{G}_\Feyn +\rmi K- G_\ret)$ is also nonnegative. 
Therefore $\widetilde{G}_\Feyn +\rmi K$ is a Feynman propagator with the properties stated in Theorem~\ref{thm:improved_Feynman}.}


\section{Real and complex quantised fields} \label{sec:algebras}

 We briefly review the construction of the complex and hermitian theories at the algebraic level to be able to justify the corresponding characterisations~\eqref{eq:covariances} and~\eqref{eq:Hcovariance} of quasifree Hadamard states and also to give an alternative proof of Corollary~\ref{cor:hermitian}.

Suppose that $\Box$ is formally self-adjoint, with no assumption that $E$ has a complex conjugation. Then the quantised theory of a complex field obeying $\Box\phi=0$ is constructed as follows. Let $\cC$ be the unital $*$-algebra with generators $\Phi(u)$, $\Phi^\star(f)$, labelled by test sections $u\in C_0^\infty(M;E^*)$ and $f\in C_0^\infty(M;E)$, subject to relations
\begin{enumerate}[label=C\arabic*] 
	\item $u\mapsto \Phi(u)$, $f\mapsto\Phi^\star(f)$ are complex-linear	
	\item $\Phi^\star(f)=\Phi(f^\star)^*$
	\item $\Phi^\star(\Box f) = 0 = \Phi(\sadj{\Box}u)$
	\item $[\Phi(u),\Phi^\star(f)] = \rmi \dlangle u, G f\drangle \one$, and $[\Phi(u),\Phi(v)]=0$
\end{enumerate}
for all $u,v\in C_0^\infty(M;E^*)$ and $f\in C_0^\infty(M;E)$. Occasionally we will write
$\cC(\Box)$ and $\Phi_\Box$, $\Phi^\star_\Box$, when it is important to emphasise the underlying operator $\Box$.

As usual, a state on $\cC$ is a linear functional $\omega:\cC\to\C$ that is normalised to $\omega(\boldsymbol{1})=1$ and is positive in the sense that $\omega(A^*A)\ge 0$ for all $A\in \cC$. 
We have already mentioned that a pair of distributions   
$W_{\Phi\Phi^\star}\in\cD'(M\times M;E\boxtimes E^*)$ and 
$W_{\Phi^\star\Phi}\in\cD'(M\times M;E^*\boxtimes E)$ obeying the relations~\eqref{eq:covariances} 
define a gauge-invariant quasifree state $\omega$ on $\cC$. Explicitly, $\omega$ is given by the identity
$$
\omega(\re^{\rmi(\Phi(u)+\Phi^\star(f))})= \re^{-\frac{1}{2}\left(
	W_{\Phi\Phi^\star}(u\otimes f)+W_{\Phi^\star\Phi}(f\otimes u) 	\right)}
$$
understood as a formal double series in $u$ and $f$, and so that
$$
W_{\Phi\Phi^\star}(u\otimes f)=\omega(\Phi(u)\Phi^\star(f)), \qquad 
W_{\Phi^\star\Phi}(f\otimes u)=\omega(\Phi^\star(f)\Phi(u))
$$
for all $f\in C_0^\infty(M;E)$, $u\in C_0^\infty(M;E^*)$.
This fact is well-known and can be found (in different notation) in e.g.,~\cites{DerezinskiGerard,Gerard}.
Using polarisation identities, linearity C1 and the commutation relations C4, the expectation value of any element of $\cC$ may be obtained. It is easily seen that the expectation value of any monomial in smeared fields vanishes unless there are equally many factors of $\Phi$ and $\Phi^\star$, which is why the state is called gauge-invariant -- the gauge transformation here being the global phase transformation $\Phi(u)\mapsto \re^{\rmi\alpha}\Phi(u)$,
$\Phi^\star(f)\mapsto  \re^{-\rmi\alpha}\Phi^\star(f)$ for $\alpha\in\R$.  As mentioned above, the choices given in~\eqref{eq:covdefs}, using Theorem~\ref{thm:DHISmod}, determine a gauge-invariant quasifree Hadamard state on $\cC$, {thus proving Corollary~\ref{cor:complex}.}

Next, suppose that $E$ has a complex conjugation $\overline{\cdot}$ with which $\Box$ commutes. Then there is a quantisation of the equation $\Box\phi=0$ in terms of a hermitian field. The algebra $\cA$ has generators $\Upsilon(f)$ labelled by $f\in C_0^\infty(M;E)$, subject to relations
\begin{enumerate}[label=H\arabic*] 
	\item  $f\mapsto\Upsilon(f)$ is complex-linear	
	\item $\Upsilon(f)=\Upsilon(\overline{f})^*$
	\item $\Upsilon(\Box f) = 0$
	\item $[\Upsilon(f),\Upsilon(h)] = \rmi \langle \overline{f}, G h\rangle \one$
\end{enumerate}
for all $f,h\in C_0^\infty(M;E)$. We will write $\cA(\Box)$ and $\Upsilon_\Box$ when necessary to indicate the underlying operator. As discussed above, a quasifree state $\omega$ on $\cA$ is determined by a distribution $ W \in\cD'(M\times M;E^*\boxtimes E^*)$ obeying conditions~\eqref{eq:Hcovariance}; in detail, $\omega$ is determined by
the identity
$$
\omega(\re^{\rmi \Upsilon(f)}) = \re^{-W(f\otimes f)/2},
$$
and one has $W(f\otimes h)=\omega(\Upsilon(f)\Upsilon(h))$. References for these standard
facts include~\cites{KhavkineMoretti:2015, KayWald:1992, DerezinskiGerard}. 

In fact, the theory of a complex field can be cast into the hermitian framework. Noting that
$E^*$ can be given a hermitian form $\langle u,v\rangle_{E^*}=\langle v^\star,u^\star\rangle_E$, we may endow $\tilde{E} = E\oplus E^*$ with a complex conjugation
$$
\overline{f\oplus u} = u^\star\oplus f^\star
$$
that commutes with the operator $\tilde{\Box}= \Box\oplus \sadj{\Box}$ acting on sections of $\tilde{E}$. Then the algebras $\cC(\Box)$ and $\cA(\tilde{\Box})$ may be identified so that
$$
\Upsilon_{\tilde{\Box}} (f\oplus u) = \Phi^\star_\Box(f) + \Phi_\Box(u)
$$
(see, e.g., Appendix~B of~\cite{Fewster:2025a}).
A general quasifree state $\omega$ of $\cA(\tilde{\Box})$ determines a gauge-invariant quasifree state of $\cC(\Box)$ provided $\omega(\Phi_\Box^\star(f)\Phi_\Box^\star(h))=
\omega(\Phi_\Box(u)\Phi_\Box(v))= 0$ for all $f,h\in C_0^\infty(M;E)$, $u,v\in C_0^\infty(M;E^*)$.
Moreover, the definition of a Hadamard state on $\cC(\Box)$ is precisely that the state be Hadamard for $\cA(\tilde{\Box})$. This correspondence is one way to prove that the two-point functions $W_{\Phi\Phi^\star}$ and $W_{\Phi^\star\Phi}$ obeying~\eqref{eq:covariances} determine a
gauge-invariant quasifree state on $\cC(\Box)$. 

There is another link between hermitian and complex fields. Suppose again that $E$ has a conjugation commuting with $\Box$. Then
$\overline{f\oplus h} = \overline{f}\oplus \overline{h}$ defines a complex conjugation on 
$E\oplus E$ with the direct sum metric, and the conjugation commutes with $\Box\oplus\Box$.  
Moreover, $\cC(\Box)$ may be identified with $\cA(\Box\oplus\Box)$ and the algebraic tensor product $\cA(\Box)\otimes \cA(\Box)$, so that
\begin{align*}
\Phi_\Box^\star(f) &= \Upsilon_{\Box\oplus\Box}\left(\frac{f}{\sqrt{2}}  \oplus \frac{-\rmi f}{\sqrt{2}} \right)
=
\frac{1}{\sqrt{2}}\left(\Upsilon_{\Box}(f)\otimes \one -\rmi \one\otimes \Upsilon_\Box(f)
\right),\nonumber\\
\Phi_\Box(u) &= \Upsilon_{\Box\oplus\Box}\left(\frac{u^T}{\sqrt{2}}  \oplus \frac{\rmi u^T}{\sqrt{2}} \right)
=
\frac{1}{\sqrt{2}}\left(\Upsilon_{\Box}(u^T)\otimes \one +\rmi \one\otimes \Upsilon_\Box(u^T)
\right),
\end{align*}
where $f\mapsto f^T$ is the linear isomorphism between (sections of) $E$ and (sections of) $E^*$ determined by $\dlangle f^T,h\drangle = \langle\overline{f},h\rangle$.
Moreover, these identifications respect the notions of Hadamard states for each theory,
because they preserve the wavefront sets of two point functions --
see Lemma~5.15 and Theorem 5.16 of~\cite{Fewster:2025a} for more general results in this direction. 
 In combination with Corollary~\ref{cor:complex}, the following lemma provides an alternative proof of Corollary~\ref{cor:hermitian}. 
\begin{lemma}\label{lem:CtoA}
	Any gauge-invariant quasifree Hadamard state of $\cC(\Box)$ with two-point distributions
	$W_{\Phi\Phi^\star}$ and $W_{\Phi^\star\Phi}$ determines a quasifree Hadamard state on $\cA(\Box)$ with two-point distribution
	$$
	W(f\otimes h) = \frac{1}{2}\left(
	W_{\Phi^\star\Phi}(f\otimes h^T+h\otimes f^T) +\rmi G_\PJ(f\otimes h)\right).  
	$$ 
\end{lemma}
\begin{proof}
	Let $\omega$ be the gauge-invariant quasifree Hadamard state on $\cC(\Box)$ determined by 
	$W_{\Phi\Phi^\star}$ and $W_{\Phi^\star\Phi}$. Under the identifications mentioned above, $\omega$ induces a state on $\cA(\Box)\otimes\cA(\Box)$ that is also Hadamard and quasifree (see Appendix~B in~\cite{Fewster:2025a} for more details). Tracing out one copy of $\cA(\Box)$ we obtain a state on $\cA(\Box)$ by $\omega'(A)=\omega(A\otimes\one)$, which is quasifree because its $n$-point functions are obtained by restriction of those of a quasifree state on $\cA(\Box\oplus\Box)$.
	The fact that $\omega'$ is Hadamard is guaranteed on general grounds because partial traces respect the Hadamard condition -- see Theorem 5.23(c) of~\cite{Fewster:2025a}. It remains to compute the two-point function. Noting that
	$$
	\Upsilon_{\Box}(f)\otimes \one = \frac{1}{\sqrt{2}}\left(\Phi_\Box^\star(f) + \Phi_\Box(f^T)\right),
	$$
	we compute
	\begin{align*}
	W(f\otimes h) &= \omega'(\Upsilon_\Box(f)\Upsilon_\Box(h)) = \omega((\Upsilon_\Box(f)\otimes\one)(\Upsilon_\Box(h)\otimes \one))  \\
	&=\frac{1}{2}
	\left(
	W_{\Phi^\star\Phi}(f\otimes h^T) + W_{\Phi\Phi^\star}(f^T\otimes h)
	\right)  \\
	&= \frac{1}{2}
	W_{\Phi^\star\Phi}(f\otimes h^T+h\otimes f^T)
	 + \frac{\rmi}{2}\dlangle f^T,G h\drangle,
	\end{align*}
	as required, where we have used~\eqref{eq:covariances} and gauge invariance of $\omega$. This completes the proof. 
	
	In fact, we may see directly that $\omega'$ fulfils the requirements~\eqref{eq:Hcovariance} to determine a quasifree state, using the last expression above to check the commutation relations and the penultimate one to check positivity. As $\WFPrime{W} \subset\WFPrime{W_{\Phi^\star\Phi}}\cup \WFPrime{W_{\Phi\Phi^\star}}\subset C_-$, it is also seen directly that $\omega'$ is Hadamard. 
\end{proof}

\section{Dirac-type operators}\label{sec:Dirac}

Let $E\to M$ be as before, with no assumption that the bundle inner product is positive definite. A formally self-adjoint operator $D$ on $C^\infty(M;E)$ is said to be of \textit{Dirac type} if its principal symbol obeys
\begin{equation*}
	\symb{D} \xxi^2 = g_{x}^{-1} (\xi, \xi) \, \one_{\End{E}} ,  
\end{equation*} 
so that $D^2$ is normally hyperbolic.
Furthermore, $D$ is of \textit{definite type}~\cites{BaerGinoux:2012,IslamStrohmaier:2020} if there is a smooth spacelike Cauchy surface $\Sigma$ for $M$ with future-pointing unit normal vector field $n$ such that 
$(f,h)\mapsto \langle \sigma_D(n^\flat)f,h\rangle$ is a positive definite sesquilinear form on $E|_\Sigma$.

The manifold $M$ need not admit a spin structure for it to support Dirac type operators, but if it does, then a conventional Dirac operator on an associated spinor bundle has definite Dirac type.  

Every Dirac type operator $D$ has Green operators $S_{\ret/\adv}$, and we set $S=S_\adv-S_\ret$. In the definite case,  we obtain an inner product $q_D$ on 
$C_0^\infty(M;E)/DC_0^\infty(M;E)$ by
$$
q_D([f],[h])=\int_\Sigma \langle\sigma_D(n^\flat)Sf,Sh\rangle \der A = \rmi \langle f, Sh\rangle,
$$ 
where $\der A$ is the metric-induced volume element on $\Sigma$.
 Moreover, the formal dual $\sadj{D}$ acting on sections of $E^*$ is also a Dirac type operator and $-\sadj{D}$ has definite type if and only if $D$ does (see~\cite{Fewster:2025a}). For consistency with~\cite{Fewster:2025a}, the conjugate-linear isomorphism between $E$ and $E^*$ induced by the fibre metric will be written as $f\mapsto f^+$ in this section, rather than $f\mapsto f^\star$.

The theory may be quantised as a fermionic QFT (see~\cites{BaerGinoux:2012,Fewster:2025a} for more detail), presented as a unital $*$-algebra $\cF$ with generators $\Psi(u)$ and $\Psi^+(f)$ 
for $f\in C_0^\infty(M;E)$, $u\in C_0^\infty(M;E^*)$ and relations 
\begin{enumerate}[label=D\arabic*] 
	\item $f\mapsto\Psi^+(f)$ and $u\mapsto \Psi(u)$ are complex-linear
	\item $\Psi(u)^*=\Psi^+(u^+)$
	\item $\Psi(\sadj{D}u) = 0 = \Psi^+(D f)$
	\item $\{\Psi(u),\Psi^+(f)\} = \rmi \dlangle u,S f\drangle\one= q_D([u^+],[f])\one$, $\{\Psi(u),\Psi(v)\} = 0$  
\end{enumerate}
for all $f\in C_0^\infty(M;E)$, $u,v\in C_0^\infty(M;E^*)$, where $\{X,Y\}=XY+YX$ is the anticommutator. If $D$ has definite type, a gauge-invariant quasifree state $\omega$ on $\cF$ is specified by 
a pair of distributions   
$W_{\Psi\Psi^+}\in\cD'(M\times M;E\boxtimes E^*)$ and 
$W_{\Psi^+\Psi}\in\cD'(M\times M;E^*\boxtimes E)$ obeying the relations  
\begin{align}\label{eq:Dirac_covariances}
	(D\otimes 1)W_{\Psi\Psi^+} &= (1\otimes \sadj{D})W_{\Psi\Psi^+} =0 \nonumber \\
	(1\otimes D)W_{\Psi^+\Psi} &= (\sadj{D}\otimes 1)W_{\Psi^+\Psi} =0 \nonumber \\
	W_{\Psi\Psi^+}(u\otimes f) + W_{\Psi^+\Psi}(f\otimes u) &= \rmi\dlangle u, Sf\drangle  \\
	W_{\Psi\Psi^+}(f^+\otimes f)& \ge 0 \nonumber \\ 
	W_{\Psi^+\Psi}(f\otimes f^+)&\ge 0\nonumber
\end{align}
for all test-sections $f\in C_0^\infty(M;E)$ and $u\in C_0^\infty(M;E^*)$, so that
\begin{align*}
	\omega(\Psi(u)\Psi^+(f))&=W_{\Psi\Psi^+}(u\otimes f), \\  
	\omega(\Psi^+(f)\Psi(u))&=W_{\Psi^+\Psi}(f\otimes u),\\
	\omega(\Psi(u)\Psi(v))&=\omega(\Psi^+(f)\Psi^+(h))=0
\end{align*}
for all $f,h\in C_0^\infty(M;E)$ and $u,v\in C_0^\infty(M;E^*)$. 
Indeed, the above conditions, together with the requirement that $\omega$ be quasifree, completely specify $\omega$ on $\cF$ and the only issue is to check positivity. For this, we refer the reader to e.g.~\cite{DerezinskiGerard}*{\S 17.2}.
Once again, gauge-invariance refers to a global $\textnormal{U}(1)$ symmetry.  

A Feynman propagator for $D$ is a Green operator $S_\Feyn$ with $\WFPrime{S_\Feyn} \subset \WFPrime{\id}\cup C_\Feyn$. In~\cite{IslamStrohmaier:2020} the following is proved
\begin{theorem}\label{thm:IS_Dirac}
	Assume that $D$ is formally self-adjoint Dirac type operator of definite type. Then there exists a Feynman propagator $S_\Feyn$ for $D$ such that 
	\begin{align*}
		-\rmi (S_\Feyn -S_\adv) \geq 0.
	\end{align*}
\end{theorem} 

In contrast to second order operators no microlocalisation is needed to prove this statement. We briefly repeat the argument here as an adaptation will be used below.
\begin{proof}
 We consider a microlocal splitting $S \sim S_+ + S_-$ with $\WFPrime{S_\pm} \subset C_\mp$. Since $\rmi S \geq 0$ we can choose this microlocal splitting directly in such a way that
 $\rmi S_\pm \geq 0$.  This is for example achieved by setting  $S_\pm =  \Psi^\pm S (\Psi^\pm)^*$ where $\Psi^\pm$
 are properly supported pseudodifferential operators inducing a microlocal splitting as follows.
 The microsupports of $\Psi^\pm$ are contained in neighborhoods of $\dot T^*_{0,\pm} M$ that do not intersect $\dot T^*_{0,\mp} M$, whereas the microsupports of $\Psi^\pm-\id$ are contained in neighborhoods of $\dot T^*_{0,\mp} M$ that do not intersect $\dot T^*_{0,\pm} M$. Such pseudo-differential operators can always be constructed by quantising symbols with these support properties.
The operator $S$ maps compactly supported sections to sections with spacelike compact support. The microlocal splitting can always be chosen in such a way that this is still true for the operators $S_\pm$ even though in general we will no longer have the strict inclusion of $\supp{S_\pm f}$ in $J(\supp{f})$. Indeed, the supports of the integral kernels of $\Psi^\pm$ can be chosen as subsets of any arbitrary small open neighborhoods of the diagonal in $M \times M$ because multiplication by a smooth function that equals one near the diagonal only changes them modulo smoothing operators.
 
 We can promote $S_\pm$ to a bi-solution by the same method as before, namely restriction to a Cauchy surface and solution of the Cauchy problem.  Here, we note that only the restriction is needed because $D$ is a first-order operator, in contrast to the four components of Cauchy data needed in the second order case. This modification changes $S_\pm$ only by 
operators that solve the bisolution Cauchy problem with zero initial data and smooth source and hence have smooth kernels.
As this also preserves non-negativity we can assume therefore without loss of generality that $S_\pm$ are bi-solutions  with $\rmi S_\pm\ge 0$. Defining
 $S_\Feyn = S_\mathrm{adv} -S_+$, one checks directly that $\WFPrime{S_\Feyn} \subset C_\Feyn$ and therefore $S_\Feyn$ is a Feynman propagator with 
 $-\rmi (S_\Feyn -S_\adv) \geq 0$.
\end{proof}

\begin{theorem}\label{thm:ISF_Dirac}
	Assume that $D$ is formally self-adjoint Dirac type operator of definite type. Then there exists a Feynman propagator $S_\Feyn$ for $D$ such that 
	\begin{align*}
		-\rmi (S_\Feyn -S_\adv) \geq 0,\\
		+\rmi (S_\Feyn -S_\ret) \geq 0 	.
	\end{align*}
\end{theorem} 

\begin{proof}
 We use the same construction as in the proof of the previous theorem, again setting $S_\Feyn = S_\adv - S_+$. The statement follows directly if we can arrange the operators $S_\pm$ in such a way
 that  $S = S_+ + S_-$ whilst retaining $\rmi S_\pm \geq 0$. To do this we choose a Cauchy surface $\Sigma$.
 We first note that the restriction of $\rmi S$ to $\Sigma \times \Sigma$  is equal to the integral kernel of the identity operator. We use here  the existence of a natural $L^2$-inner product on the Cauchy data space 
$C^\infty_0(\Sigma, E|_\Sigma)$ and the associated Hilbert space $L^2(\Sigma, E|_\Sigma)$ due to the definite type assumption on $D$.
 
 As in the previous theorem we can find a microlocal splitting  $\rmi S\sim \rmi S_++ \rmi S_-$, so that the restriction of $\rmi S_+$ and $\rmi S_-$ to $\Sigma \times \Sigma$ yields two positive kernels $Q_+$ and $Q_-$ with the property that $\id + k = Q_+ + Q_-$ for some smooth kernel $k$ (here, we have not promoted $S_\pm$ to bisolutions). Furthermore, since $S_\pm$ map compactly supported sections to spacelike compactly supported sections, each $Q_\pm$ is a properly supported pseudodifferential operator.
 By the domination lemma, Lemma \ref{lem:domination}, the kernel $k$ can be dominated by a positive smooth properly supported kernel and we can thus modify $Q_\pm$ so that $Q_+ + Q_-= \id + K$, where $K$ is a smooth non-negative kernel that is properly supported.
 By Lemma \ref{rootlemma} the operator $(\id+K)^{-\frac{1}{2}}$ differs from the identity map by a smooth kernel. Therefore
 $\tilde Q_\pm = (\id+K)^{-\frac{1}{2}} Q_\pm (\id+K)^{-\frac{1}{2}}$ is a splitting $\tilde Q_+ + \tilde Q_- =\id$, with $\tilde{Q}_\pm-Q_\pm$ being smoothing operators, noting that $Q_\pm$ is properly supported.
 Extending $\tilde Q_\pm$ to bisolutions on $M \times M$ we obtain modifications of $\rmi S_\pm$ by smoothing operators so that
 $\rmi S = \rmi S_+ + \rmi S_-$, retaining positivity. Note that $(\id+K)^{-\frac{1}{2}}$ is not necessarily properly supported and neither  are $S_\pm$ properly supported in general.
  \end{proof}

Setting $$W_{\Psi\Psi^+}(u\otimes f) = -\rmi\langle u^+,(S_\Feyn-S_\adv) f\rangle,\qquad 
W_{\Psi^+\Psi}(f\otimes u)={+}\rmi\langle u^+,(S_\Feyn-S_\ret) f\rangle,$$
for $f\in C_0^\infty(M;E)$, $u\in C_0^\infty(M;E^*)$, one easily checks that~\eqref{eq:Dirac_covariances} are satisfied and that 
$$
\WFPrime{W_{\Psi\Psi^+}} \cup \WFPrime{W_{\Psi^+\Psi}} \subset {C_-},
$$
which shows that $\omega$ is Hadamard. 
Summarising, we have the following statement that completes and makes explicit Theorem~1.5 of~\cite{IslamStrohmaier:2020} and is used to prove Theorem~5.7 of~\cite{Fewster:2025a}.
\begin{corollary}
	Let $D$ be any formally self-adjoint Dirac type operator of definite type on a smooth finite-rank complex hermitian vector bundle $E\to M$ over a globally hyperbolic spacetime $M$. Then the quantised fermionic field with equation of motion operator $D$ admits a gauge-invariant quasifree Hadamard state, whose two-point functions are determined by the Feynman propagator provided by Theorem~\ref{thm:ISF_Dirac}.
\end{corollary}

Lastly, suppose $E$ is equipped with a skew complex conjugation, i.e., a conjugate-linear fibrewise involution such that 
$$
\langle \Gamma f,\Gamma h\rangle = -\langle h,f\rangle
$$ 
for all $f,h\in C^\infty(M;E)$. Suppose that $D$ is formally self-adjoint Dirac type and of definite type, which commutes with $\Gamma$. Then  $D$ may be quantised as a hermitian \textit{Majorana} field theory, in terms of a unital $*$-algebra $\cM$ 
generated by $\Xi(f)$, $f\in C_0^\infty(M;E)$ satisfying
\begin{enumerate}[label=M\arabic*] 
	\item $f\mapsto\Xi(f)$ is complex-linear
	\item $\Xi(f)^*=\Xi(\Gamma{f})$
	\item $\Xi(Df) = 0$
	\item $\{\Xi(f),\Xi(h)\} = \rmi \langle \Gamma{f},S h\rangle\one = q_D([\Gamma f],[h])\one$
\end{enumerate}
for all $f,h\in C_0^\infty(M;E)$ (cf.\ e.g.,~\cite{SahlmannVerch:2000RMP}).
A quasifree state $\omega$ on $\cM$ is determined by a distribution $W\in \cD'(M\times M;E^*\boxtimes E^*)$ obeying
\begin{align}\label{eq:Majorana_covariance}
	W\circ (D\otimes 1) &= W\circ (1\otimes {D}) = 0 \nonumber \\
	W(f\otimes h) +W(h\otimes f) &=  \rmi \langle \Gamma{f},S h\rangle \\
	W(\Gamma{f}\otimes f) & \ge 0,\nonumber
\end{align}
 so that $\omega(\Xi(f)\Xi(h))= W(f\otimes h)$ for all $f,h\in C_0^\infty(M;E)$.
By similar arguments to those used for the hermitian bosonic field, one may prove:
\begin{theorem}\label{thm:Majorana}
	Assume in addition to the hypotheses of Theorem~\ref{thm:IS_Dirac} 
	that $D$ commutes with a skew complex conjugation {$\Gamma$} on $E$. 
	Then there is a Feynman propagator $\widetilde{S}_\Feyn$ for $D$ with the properties of {Theorem~\ref{thm:ISF_Dirac}} and 
	\begin{equation}\label{eq:skew}
	\langle \Gamma f, \widetilde{S}_\Feyn h\rangle =-\langle \Gamma h, \widetilde{S}_\Feyn f\rangle.
	\end{equation}
	Explicitly,
	$$
		\langle \Gamma f, \widetilde{S}_\Feyn h\rangle =
		\frac{1}{2}\left(\langle \Gamma f, S_\Feyn h\rangle -
			\langle \Gamma h, S_\Feyn  f\rangle\right) 
	$$
	where $S_\Feyn$ is a {Feynman} propagator furnished by Theorem~\ref{thm:ISF_Dirac}.
\end{theorem}
\begin{proof}
	Using the fact that $\Gamma$ is a skew conjugation and the properties of $S_\Feyn$, it is easily seen that $\widetilde{S}_\Feyn$ is a Green operator. The kernel distribution is
	$$
	\widetilde{S}_\Feyn(u\otimes f) = \frac{1}{2}\left( S_\Feyn(u\otimes f) - S_\Feyn((\Gamma f)^+\otimes (\Gamma u^+)\right);
	$$
	consequently, one has the wave-front set inclusion
	$$
	\WFPrime{\widetilde{S}_\Feyn}\subset \WFPrime{S_\Feyn}\cup \cR \WFPrime{S_\Feyn},
	$$
	where $\cR(x,\xi;x',\xi')=(x',-\xi';x,-\xi)$. As $\WFPrime{S_\Feyn}$ is $\cR$-invariant, it follows that $\WFPrime{\widetilde{S}_\Feyn}\subset \WFPrime{\id}\cup C_\Feyn$, {so $S_\Feyn$ is a Feynman propagator. Finally, we may compute
	\[
	-\rmi\langle f,(\widetilde{S}_\Feyn-S_\adv) f\rangle = -\frac{\rmi}{2}\langle f,(S_\Feyn-S_\adv) f\rangle +  \frac{\rmi}{2}\langle \Gamma f,(S_\Feyn-S_\ret) \Gamma f\rangle +
	\frac{\rmi}{2}(\langle f,S_\adv f\rangle - \langle \Gamma S_\ret \Gamma f,f\rangle)
	\]
	using the properties of the skew-conjugation. The first two terms on the right-hand side are nonnegative by Theorem~\ref{thm:ISF_Dirac} and the last one vanishes because $\Gamma$ commutes with $S_\ret$, which is the formal adjoint of $S_\adv$. Thus $-\rmi (\widetilde{S}_\Feyn-S_\adv)\ge 0$ and an analogous calculation shows that $+\rmi (\widetilde{S}_\Feyn-S_\ret)\ge 0$ as well.
	}
\end{proof}
\begin{corollary}
	Let $D$ be any formally self-adjoint Dirac type operator of definite type on a smooth finite-rank complex hermitian vector bundle $E\to M$ over a globally hyperbolic spacetime $M$, and suppose $D$ commutes with a skew complex conjugation on $E$. Then the quantised Majorana theory $\cM(D)$ admits a quasifree Hadamard state. Explicitly, the two-point function can be given as either 
	{ $$ 
	W(f\otimes h) = -\rmi\langle \Gamma f,(S_\Feyn-S_\adv)h\rangle
	$$}
	using the $S_\Feyn$ from Theorem~\ref{thm:Majorana}, or equivalently by the formula
	$$
		W(f\otimes h)= \frac{1}{2}W_{\Psi^+\Psi}(f\otimes (\Gamma h)^+ - 
	h\otimes (\Gamma f)^+)+ \frac{\rmi}{2}\dlangle (\Gamma f)^+,Sh\drangle.
	$$
	where $W_{\Psi^+\Psi}$ is a two-point function for a gauge-invariant quasifree Hadamard state of
	the Dirac theory $\cF(D)$.
\end{corollary}
\begin{proof}
	The first and third parts of~\eqref{eq:Majorana_covariance} are clear, so it only remains to check the anticommutation relations and the equivalent formula. Using the property~\eqref{eq:skew} established in Theorem~\ref{thm:Majorana}, 
	$$
	W(h\otimes f) = -\rmi\langle \Gamma h,(S_\Feyn-S_\adv)f\rangle = 
	\rmi \langle \Gamma f,S_\Feyn h\rangle + \rmi \langle \Gamma h,S_\adv f\rangle =
	\rmi \langle \Gamma f,(S_\Feyn-S_\ret)h\rangle
	$$
	by skew-symmetry of $\Gamma$ and the fact that $S_\adv$ commutes with $\Gamma$ and has formal adjoint $S_\ret$. The anticommutation relations follow. 
	For the alternative formula, note that 
	$$
	W_{\Psi^+\Psi}(f\otimes (\Gamma h)^+ ) = \rmi\langle \Gamma h,(S_\Feyn-S_\ret)f\rangle = W(f\otimes h)
	$$
	by the equation above, while 
	$$
	\rmi \dlangle (\Gamma f)^+,Sh\drangle - W_{\Psi^+\Psi}(h\otimes (\Gamma f)^+) =
	\rmi\langle \Gamma f,Sh\rangle - \rmi \langle\Gamma f,(S_\Feyn-S_\ret)h\rangle = W(f\otimes h).
	$$
	The proof is completed by taking an average.
\end{proof}

\subsection*{Acknowledgements} 
We thank Elmar Schrohe for useful conversations.
The work of CJF is partly supported by EPSRC Grant EP/Y000099/1 to the University of York.  

For the purpose of open access, the authors have applied a creative commons attribution (CC BY) licence to any author accepted manuscript version arising.

\subsection*{Author declaration} The authors have no conflicts of interest to declare. The paper did not require ethical approval. The authors contributed equally and each more than 50\%.

\subsection*{Data availability statement} 

No data was generated for this paper.

\appendix
\section{Notational comparison with~\cite{IslamStrohmaier:2020} and~\cite{Fewster:2025a}}\label{appx:notation}
 
	For ease of comparison, we indicate the main similarities and differences between our conventions and notation and those of~\cites{IslamStrohmaier:2020,Fewster:2025a}.
	
	The sign conventions used here are mostly the same as in~\cite{Fewster:2025a}, except that we (like~\cite{IslamStrohmaier:2020}) use the conventional sign choice when defining the Fourier transform, while~\cite{Fewster:2025a} has a reversed sign in the exponent, thus reversing the sign of covectors in wavefront sets. Thus, for example, the Hadamard condition here for a two-point function of a hermitian field is $\WF{W}\subset \dot{T}_{0,-}M\times \dot{T}_{0,+}M$ rather than~\cite{Fewster:2025a}'s $\WF{W}\subset \dot{T}_{0,+}M\times \dot{T}_{0,-}M$. By contrast,~\cite{IslamStrohmaier:2020} uses a mostly-positive metric signature. Thus `musical isomorphisms' are related by a relative sign. Nonetheless, the definition of the principal symbol of an operator is common, as is the class of normally hyperbolic operators, because the minus sign in the relation between the principal symbol and metric, here and in~\cite{Fewster:2025a}, is absent in~\cite{IslamStrohmaier:2020}. Similarly, the definitions of Dirac type and definite type operators agree (neither~\cite{Fewster:2025a} nor~\cite{IslamStrohmaier:2020} discusses Majorana type operators). The advanced and retarded Green functions are defined in the same way as in~\cites{Fewster:2025a,IslamStrohmaier:2020}.
	
The term `formally self-adjoint' here corresponds to `formally hermitian' in~\cite{Fewster:2025a}.

Finally, we have followed~\cite{IslamStrohmaier:2020} in freely identifying densities and functions using the metric density, whereas~\cite{Fewster:2025a} maintains the distinction.



\begin{bibdiv}
\begin{biblist}

\bib{Baer:2015}{article}{
	AUTHOR = {B{\"a}r, C.},
	TITLE = {Green-hyperbolic operators on globally hyperbolic spacetimes},
	JOURNAL = {Comm. Math. Phys.},
	FJOURNAL = {Communications in Mathematical Physics},
	VOLUME = {333},
	YEAR = {2015},
	NUMBER = {3},
	PAGES = {1585--1615},
	ISSN = {0010-3616},
	MRCLASS = {35L30 (35L15 35L45 35R01 58J45)},
	MRNUMBER = {3302643},
	MRREVIEWER = {Herbert Koch},
	doi = {10.1007/s00220-014-2097-7},
	noURL = {http://dx.doi.org/10.1007/s00220-014-2097-7},
}


\bib{BaerGinoux:2012}{incollection}{
	AUTHOR = {B\"{a}r, C.},
	author = {Ginoux, N.},
	TITLE = {Classical and quantum fields on {L}orentzian manifolds},
	BOOKTITLE = {Global differential geometry},
	SERIES = {Springer Proc. Math.},
	VOLUME = {17},
	PAGES = {359--400},
	PUBLISHER = {Springer, Heidelberg},
	YEAR = {2012},
	MRCLASS = {81T08 (53C27 53C50 81T20)},
	MRNUMBER = {3289848},
	MRREVIEWER = {Kotik K. Lee},
	DOI = {10.1007/978-3-642-22842-1\_12},
	URL = {https://doi.org/10.1007/978-3-642-22842-1_12},
}

\bib{dergas2025prop}{article}{
      title={Propagators in curved spacetimes from operator theory}, 
      author={Dereziński, J.},
      author={Ga\ss, C.},
      year={2025},
      eprint={arXiv:2409.03279},
      archivePrefix={arXiv},
      primaryClass={math-ph},
      url={https://arxiv.org/abs/2409.03279}, 
}

\bib{DerezinskiGerard}{book}{
	AUTHOR = {Derezi\'nski, J.},
	author = {G\'erard, C.},
	TITLE = {Mathematics of quantization and quantum fields},
	SERIES = {Cambridge Monographs on Mathematical Physics},
	PUBLISHER = {Cambridge University Press, Cambridge},
	YEAR = {2013},
	PAGES = {xii+674},
	ISBN = {978-1-107-01111-3},
	MRCLASS = {81-02 (81S05 81S10 81T70)},
	MRNUMBER = {3060648},
	DOI = {10.1017/CBO9780511894541},
	URL = {https://doi.org/10.1017/CBO9780511894541},
}

\bib {DuiHoer_FIOii:1972}{article}{
	AUTHOR = {Duistermaat, J. J.},
	author = {H{\"o}rmander, L.},
	TITLE = {Fourier integral operators. {II}},
	JOURNAL = {Acta Math.},
	FJOURNAL = {Acta Mathematica},
	VOLUME = {128},
	YEAR = {1972},
	NUMBER = {3-4},
	PAGES = {183--269},
	ISSN = {0001-5962,1871-2509},
	MRCLASS = {58G15 (35S05 47G05)},
	MRNUMBER = {388464},
	MRREVIEWER = {Yu.\ V.\ Egorov},
	DOI = {10.1007/BF02392165},
	URL = {https://doi.org/10.1007/BF02392165},
}



\bib{Fewster:2025a}{article}{ 
		AUTHOR = {Fewster, C. J.},
		TITLE = {Hadamard states for decomposable {G}reen-hyperbolic operators},
		JOURNAL = {Comm. Math. Phys.},
		FJOURNAL = {Communications in Mathematical Physics},
		VOLUME = {407},
		YEAR = {2026},
		NUMBER = {1},
		PAGES = {Paper No. 14, 57},
		ISSN = {0010-3616,1432-0916},
		MRCLASS = {81T20 (58J45)},
		MRNUMBER = {4998897},
		DOI = {10.1007/s00220-025-05512-1},
		URL = {https://doi.org/10.1007/s00220-025-05512-1},
		eprint={arXiv:2503.12537},
	}
%


\bib{FullingSweenyWald:1978}{article}{
	AUTHOR = {Fulling, S. A.},
	author={Sweeny, M.},
	author={Wald, R. M.},
	TITLE = {Singularity structure of the two-point function quantum field
		theory in curved spacetime},
	JOURNAL = {Comm. Math. Phys.},
	FJOURNAL = {Communications in Mathematical Physics},
	VOLUME = {63},
	YEAR = {1978},
	NUMBER = {3},
	PAGES = {257--264},
	ISSN = {0010-3616,1432-0916},
	MRCLASS = {81M05 (83C45)},
	MRNUMBER = {513903},
	MRREVIEWER = {B.\ K.\ Nayak},
	URL = {http://projecteuclid.org/euclid.cmp/1103904566},
	doi = {10.1007/BF01196934},
}

\bib{Gerard}{book}{
	AUTHOR = {G\'erard, C.},
	TITLE = {Microlocal analysis of quantum fields on curved spacetimes},
	SERIES = {ESI Lectures in Mathematics and Physics},
	PUBLISHER = {EMS Publishing House, Berlin},
	YEAR = {[2019] \copyright 2019},
	PAGES = {viii+220},
	ISBN = {978-3-03719-094-4},
	MRCLASS = {81T20 (35A27 35Q40 58J40 58J47 81T28)},
	MRNUMBER = {3972066},
	DOI = {10.4171/094},
	URL = {https://doi.org/10.4171/094},
}

\bib{GerardOulghaziWrochna:2017}{article}{
	AUTHOR = {G\'erard, C.},
	author= {Oulghazi, O.},
	author= {Wrochna, M.},
	TITLE = {Hadamard states for the {K}lein-{G}ordon equation on
		{L}orentzian manifolds of bounded geometry},
	JOURNAL = {Comm. Math. Phys.},
	FJOURNAL = {Communications in Mathematical Physics},
	VOLUME = {352},
	YEAR = {2017},
	NUMBER = {2},
	PAGES = {519--583},
	ISSN = {0010-3616,1432-0916},
	MRCLASS = {83C47 (35L10 35Q75 35R01 81T20 83C45 83E30)},
	MRNUMBER = {3627405},
	MRREVIEWER = {Mark\ D.\ Roberts},
	DOI = {10.1007/s00220-017-2847-4},
	URL = {https://doi.org/10.1007/s00220-017-2847-4},
}

\bib{Guenther_AP_1988}{book}{
	title=		{Huygens' Principle and Hyperbolic Equations},
	author=		{G\"{u}nther, P.},
	isbn=		{978-0-12-307330-3},
	series=		{Prespective in Mathematics},
	volume=		{5},
		doi=		{https://doi.org/10.1016/C2013-0-10776-3},
	year=		{1988},
	publisher=	{Academic Press}, 
	address= 	{USA}
}


\bib{Hadamard_ActaMath_1908}{article}{
	author=		{Hadamard, J.},
	doi=		{10.1007/BF02415449},
	journal=	{Acta Math.},
	pages=		{333-380},
	publisher=	{Institut Mittag-Leffler},
	title=		{Th\'{e}orie des \'{e}quations aux d\'{e}riv\'{e}es partielles lin\'{e}aires hyperboliques et du probl\`{e}me de Cauchy},
	url=		{https://doi.org/10.1007/BF02415449},
	volume=		{31},
	year=		{1908}
}

\bib{Hadamard_Dover_2003}{book}{
	title=		{Lectures on {C}auchy's problem in linear partial differential equations},
	author=		{Hadamard, J.},
	isbn=		{},
	series=		{},
	volume=		{},
	doi=		{},
	year=		{2003},
	publisher=	{Dover}, 
	address=	{NY}
}



\bib{Hoermander_Nice_1970}{article}{
	AUTHOR=	{H\"{o}rmander, L.}, 
	Editor=	{},
	title=		{Linear Differential Operators},
	conference={
		title=	{Actes, Congr\'{e}s intern. math,},
		date=		{1970},
		Address=	{Nice, France}
	},
	Pages=		{121-133},
	Publisher=	{},
	Month=		{1-10 Sep.},
}

\bib{Hoermander_ActaMath_1971}{article}{
	author=		{H\"{o}rmander, L.},
	doi=		{10.1007/BF02392052},
	fjournal=	{Acta Mathematica},
	journal=	{Acta Math.},
	pages=		{79-183},
	title=		{Fourier integral operators. {I}},
	volume=		{127},
	year=		{1971}
}

\bib{Hoermander_Springer_2003}{book}{
 	author=		{H\"{o}rmander, L.},
 	title=		{The Analysis of Linear Partial Differential Operators {I}: Distribution Theory and Fourier Analysis},
 	series=		{Classics in Mathematics},
 	volume=		{},
 	edition=	{2},
 	publisher=	{Springer-Verlag},
 	address=	{Berlin, Heidelberg},
 	year=		{2003},
 doi=		{10.1007/978-3-642-61497-2},
 	url=		{https://www.springer.com/gb/book/9783540006626}
}


\bib{Hollands_CMP_2001}{article}{
 	author=	{Hollands, S.},
 	author=	{Wald, R. M.},
 	title=	{Local {W}ick polynomials and time ordered products of quantum fields in curved spacetime},
 	journal={Commun. Math. Phys.},
 	year=	{2001},
 	volume=	{223},
 	number=	{2},
 	pages=	{289-326},
 doi=		{10.1007/s002200100540},
 	url=	{http://dx.doi.org/10.1007/s002200100540},
 	eprint= {arXiv:0103074 [gr-qc]}
}

\bib{Hollands_CMP_2002}{article}{
 	author=	{Hollands, S.},
 	author=	{Wald, R. M.},
 	title=	{Existence of local covariant time ordered products of quantum fields in curved spacetime},
 	journal={Commun. Math. Phys.},
 	year=	{2002},
 	volume=	{231},
 	number=	{2},
 	pages=	{309-345},
  doi=		{10.1007/s00220-002-0719-y},
 	url=	{http://dx.doi.org/10.1007/s00220-002-0719-y},
 	eprint=	{arXiv:0111108 [gr-qc]}
}



\bib{IslamStrohmaier:2020}{article}{
	doi = {10.4310/CAG.241204020919},
	url = {https://doi.org/10.4310/CAG.241204020919},	
	author = {Islam, O.},
	author = {Strohmaier, A.},	
	keywords = {Analysis of PDEs (math.AP), Mathematical Physics (math-ph), Differential Geometry (math.DG), FOS: Mathematics, FOS: Mathematics, FOS: Physical sciences, FOS: Physical sciences, 35S30, 58J40, 81T20},	
	title = {On microlocalization and the construction of {F}eynman propagators for normally hyperbolic operators},
	journal = {Communications in Analysis and Geometry},
	volume={32},
	pages={1811-1883},
	eprint = {arXiv:2012.09767},
	year = {2024},	
	copyright = {arXiv.org perpetual, non-exclusive license}
}



\bib{Junker_RMP_1996}{article}{
	author=	{Junker, W.},
	title=	{Hadamard states, adiabatic vacua and the construction of physical states for scalar quantum fields on curved space-time},
	journal={Rev. Math. Phys.},
	volume=	{8},
	pages=	{1091-1159},
	year=	{1996},
	note=	{Erratum: Rev.Math.Phys. \textbf{14}, 511-517 (2002)}
}

\bib{Junker_AHP_2002}{article}{
	author=	{Junker, W.},
	author=	{Schrohe, E.},
	title=	{Adiabatic vacuum states on general space-time manifolds: Definition, construction, and physical properties},
	eprint=	{arXiv:0109010 [math-ph]}, 
	journal={Ann. Henri Poincar\'{e}},
	volume=	{3},
	pages=	{1113-1182},
	year=	{2002}
}

\bib{KayWald:1992}{article}{
	AUTHOR = {Kay, B. S.},
	AUTHOR = {Wald, R. M.},
	TITLE = {Theorems on the uniqueness and thermal properties of
		stationary, nonsingular, quasifree states on spacetimes with a
		bifurcate {K}illing horizon},
	JOURNAL = {Phys. Rep.},
	FJOURNAL = {Physics Reports. A Review Section of Physics Letters},
	VOLUME = {207},
	YEAR = {1991},
	NUMBER = {2},
	PAGES = {49--136},
	ISSN = {0370-1573,1873-6270},
	MRCLASS = {81T20 (46L60 81T05 83C47)},
	MRNUMBER = {1133130},
	MRREVIEWER = {Stephen\ J.\ Summers},
	DOI = {10.1016/0370-1573(91)90015-E},
	URL = {https://doi.org/10.1016/0370-1573(91)90015-E},
}


\bib{KhavkineMoretti:2015}{incollection}{
	AUTHOR = {Khavkine, I.},
	author = {Moretti, V.},
	TITLE = {Algebraic {QFT} in curved spacetime and quasifree {H}adamard
		states: an introduction},
	BOOKTITLE = {Advances in algebraic quantum field theory},
	SERIES = {Math. Phys. Stud.},
	PAGES = {191--251},
	PUBLISHER = {Springer, Cham},
	YEAR = {2015},
	ISBN = {978-3-319-21352-1; 978-3-319-21353-8},
	MRCLASS = {81T20},
	MRNUMBER = {3409590},
}

\bib{Lewandowski_JMP_2022}{article}{
	author=	{Lewandowski, M.},
	title=	{Hadamard states for bosonic quantum field theory on globally hyperbolic spacetimes},
	journal={J. Math. Phys.},
	volume=	{63},
	number=	{1},
	pages=	{013501},
	year=	{2022},
	eprint=	{arXiv:2008.13156 [math-ph]}
}

\bib{Moretti:2008}{article}{
	AUTHOR = {Moretti, V.},
	TITLE = {Quantum out-states holographically induced by asymptotic
		flatness: invariance under spacetime symmetries, energy
		positivity and {H}adamard property},
	JOURNAL = {Comm. Math. Phys.},
	FJOURNAL = {Communications in Mathematical Physics},
	VOLUME = {279},
	YEAR = {2008},
	NUMBER = {1},
	PAGES = {31--75},
	ISSN = {0010-3616,1432-0916},
	MRCLASS = {83C47 (81T20 83C30)},
	MRNUMBER = {2377628},
	MRREVIEWER = {Gerald\ Hofmann},
	noDOI = {10.1007/s00220-008-0415-7},
	URL = {https://doi.org/10.1007/s00220-008-0415-7},
	eprint         = {gr-qc/0610143},
	archivePrefix  = {arXiv},
}



\bib{Radzikowski_CMP_1996}{article}{
	author=		{Radzikowski, M. J.},
	journal=	{Commun. Math. Phys.},
	number=		{},
	pages=		{529-553},
	publisher=	{Springer},
	title=		{Micro-local approach to the {H}adamard condition in quantum field theory on curved space-time},
	url=		{http://projecteuclid.org/euclid.cmp/1104287114},
	volume=		{179},
	year=		{1996}
}

\bib{Riesz_ActaMath_1949}{article}{
	author=	{Riesz, M.},
	journal=	{Acta Math.},
	pages=		{1-222},
	publisher=	{Institut Mittag-Leffler},
	title=		{L'int\'{e}grale de Riemann-Liouville et le probl\`{e}me de {C}auchy},
	url=		{https://doi.org/10.1007/BF02395016},
	volume=		{81},
	year=		{1949}
}

\bib{Riesz_CPAM_1960}{article}{
	author=		{Riesz, M.},
	title=		{A geometric solution of the wave equation in space-time of even dimension},
	journal=	{Comm. Pure Appl. Math.},
	volume=		{13},
	number=		{3},
	pages=		{329-351},
	url=		{https://onlinelibrary.wiley.com/doi/abs/10.1002/cpa.3160130302},
	year=		{1960}
}


\bib{SahlmannVerch:2000RMP}{article}{
	author         = {Sahlmann, H.},
	author 			= {Verch, R.},
	title          = {{Microlocal spectrum condition and Hadamard form for
			vector valued quantum fields in curved space-time}},
	journal        = {Rev. Math. Phys.},
	volume         = {13},
	pages          = {1203-1246},
	doi           = {10.1142/S0129055X01001010},
	year           = {2001},
	eprint         = {math-ph/0008029},
}

\bib{MR1741419}{book}{
   author={Schaefer, H. H.},
   author={Wolff, M. P.},
   title={Topological vector spaces},
   series={Graduate Texts in Mathematics},
   volume={3},
   edition={2},
   publisher={Springer-Verlag, New York},
   date={1999},
   pages={xii+346},
}

\bib{MR3611021}{article}{
   author={Vasy, A.},
   title={On the positivity of propagator differences},
   journal={Ann. Henri Poincar\'e},
   volume={18},
   date={2017},
   number={3},
   pages={983--1007},
   issn={1424-0637},
   review={\MR{3611021}},
   doi={10.1007/s00023-016-0527-0},
}

\bib{VasyWrochna:2018}{article}{
	AUTHOR = {Vasy, A.},
	author = {Wrochna, M.},
	TITLE = {Quantum fields from global propagators on asymptotically
		{M}inkowski and extended de {S}itter spacetimes},
	JOURNAL = {Ann. Henri Poincar\'e},
	FJOURNAL = {Annales Henri Poincar\'e. A Journal of Theoretical and
		Mathematical Physics},
	VOLUME = {19},
	YEAR = {2018},
	NUMBER = {5},
	PAGES = {1529--1586},
	ISSN = {1424-0637,1424-0661},
	MRCLASS = {58J45 (35L10 35R01 81Q05 81T20 83C47)},
	MRNUMBER = {3784921},
	MRREVIEWER = {Jonathan\ Michael\ Blackledge},
	DOI = {10.1007/s00023-018-0650-1},
	URL = {https://doi.org/10.1007/s00023-018-0650-1},
}
\end{biblist}
\end{bibdiv}
\end{document}